\documentclass[11pt,a4paper]{article}

\usepackage[numbers]{natbib}
\usepackage{xspace}
\usepackage[margin=.9in]{geometry}
\usepackage{pstricks}
\usepackage{comment}
\usepackage{dsfont}
\usepackage{tikz}
\usepackage{tikz-dependency}

\usepackage{amsmath,amssymb,amsthm,amsfonts,latexsym,bbm,xspace,graphicx,float,mathtools,mathdots}
\usepackage{braket,caption,subcaption,ellipsis,xcolor,textcomp}
\usepackage[colorlinks,linkcolor=red,citecolor=blue]{hyperref}
\usepackage{enumitem}
\usepackage[ruled,linesnumbered]{algorithm2e}
\IncMargin{-\parindent}
\usepackage{algpseudocode}
\usepackage{listings}
\usepackage[nocomma, short]{optidef}

\usepackage{thmtools}
\usepackage{thm-restate}

\allowdisplaybreaks

\newtheorem{theorem}{Theorem}[section]
\declaretheorem[name=Theorem,sibling=theorem]{thm-restate}

\newtheorem{lemma}[theorem]{Lemma}

\newtheorem{proposition}[theorem]{Proposition}

\newtheorem{corollary}[theorem]{Corollary}

\theoremstyle{definition}
\newtheorem{definition}[theorem]{Definition}

\renewcommand{\Pr}{\mathop{\bf Pr\/}}

\newcommand{\Ex}[1]{{\mathbb{E}}\left[#1\right]}

\newcommand{\classP}{\textsf{P}}
\newcommand{\classNP}{\textsf{NP}}
\newcommand{\classAPX}{\textsf{APX}}

\newcommand{\DTIME}{\textsf{DTIME}}

\newcommand{\eps}{\varepsilon}

\newcommand{\ot}{\otimes}

\newcommand{\calD}{\mathcal{D}}

\newcommand{\calL}{\mathcal{L}}

\newcommand{\calN}{\mathcal{N}}

\newtheorem{problem}{Problem}

\newcommand{\NN}[0]{\mathbb{N}}
\newcommand{\RR}[0]{\mathbb{R}}
\newcommand{\tSigma}{\tilde{\Sigma}}
\newcommand{\OPT}{\mathrm{OPT}}
\newcommand{\ALG}{\mathrm{ALG}}

\newcommand{\D}{\displaystyle}

\renewcommand{\Pr}[1]{\mbox{\rm\bf Pr}\left[#1\right]} 
\newcommand{\growingmid}{\mathrel{}\middle|\mathrel{}}

\newcommand{\IS}{\textsc{Independent Set}}
\newcommand{\VC}{\textsc{Vertex Cover}}
\newcommand{\MkVC}{\textsc{Max-k-Vertex-Cover}}
\newcommand{\BVE}{\textsc{Bipartite Vertex Expansion}}
\newcommand{\CBVE}{\textsc{Colored Bipartite Vertex Expansion}}
\newcommand{\PART}{\textsc{Partition}}

\newcommand{\sdpopt}{\mathcal{V}_{SDP}}

\title{Algorithmic Persuasion with Evidence}

\author{Martin Hoefer\thanks{RWTH Aachen University, Germany. mhoefer@cs.rwth-aachen.de} \and Pasin Manurangsi\thanks{Google Research, USA. pasin@google.com} \and Alexandros Psomas\thanks{Purdue University, USA. apsomas@purdue.edu}
}
\date{}

\begin{document}

\maketitle

\begin{abstract}%
In a game of persuasion with evidence, a sender has private information. By presenting evidence on the information, the sender wishes to persuade a receiver to take a single action (e.g., hire a job candidate, or convict a defendant). The sender's utility depends solely on whether or not the receiver takes the action. The receiver's utility depends on both the action and the sender's private information. 

We study three natural variations. First, we consider the problem of computing an equilibrium of the game without commitment power. Second, we consider a persuasion variant, where the sender commits to a signaling scheme and the receiver, after seeing the evidence, takes the action or not. Third, we study a delegation variant, where the receiver first commits to taking the action if being presented certain evidence, and the sender presents evidence to maximize the probability the action is taken. 
We study these variants through the computational lens, and give hardness results, optimal approximation algorithms, and polynomial-time algorithms for special cases. Among our results is an approximation algorithm that rounds a semidefinite program that might be of independent interest, since, to the best of our knowledge, it is the first such approximation algorithm in algorithmic economics.
\end{abstract}%

\section{Introduction}


Persuasion is a fundamental challenge arising in diverse areas such as recommendation problems in the Internet, consulting and lobbying, employee hiring. Persuasion problems occupy a central role in economics and received significant interest over the last two decades. A prominent approach is \emph{persuasion with evidence} as introduced by Glazer and Rubinstein~\cite{GlazerR04,GlazerR06}. 
In this problem, a sender wishes to persuade a receiver to take a single action by presenting evidence. The sender's utility depends solely on whether or not the action is taken, while the receiver's utility depends on both the action as well as the sender's private information. Consider, for example, a prosecutor trying to convince a judge that a defendant is guilty and should be convicted, or a job candidate trying to convince a company that she has the best qualifications and should be hired. How should these pairs of agents interact? 

The literature on persuasion games in economics and game theory is vast; see Sobel~\cite{sobel2013giving} for a survey. In sharp contrast, very little is known about \emph{computation} in this domain, especially for the persuasion problem with evidence. How does the restriction to evidence impact the computational complexity of persuasion strategies? Our main contribution of this paper is to initiate the systematic study of \emph{persuasion with evidence though a computational lens}.

We examine three natural model variants that arise from the power to commit to certain behavior. If there is no commitment power, the scenario is an extensive-form game. We prove that finding a subgame-perfect equilibrium is always possible in polynomial time. However, the sender and the receiver can significantly improve their utility when they enjoy commitment power. 

If the sender has commitment power, then she can commit in advance which evidence is presented in each possible instantiation of her private information, and the receiver seeing the evidence then takes the action or not. We refer to this situation as \emph{constrained persuasion}, since the sender with commitment power wants to persuade the rational receiver to take the action. The sender is constrained to providing concrete evidence instead of just making a recommendation as is the case in the so called Bayesian persuasion paradigm~\cite{kamenica2011bayesian}. Constrained persuasion is a natural model in the example of prosecutor and judge, where the prosecutor (sender) with private information would first present evidence before the judge (receiver) makes a decision. Although this scenario seems structurally rather simple, we show that the sender's task in constrained persuasion can in general become computationally (highly) intractable. Unless $\classP = \classNP$, optimal persuasion is hard to approximate within a polynomial factor of the input size. However, many persuasion scenarios exhibit a natural condition that we term ``global signal'': At least one signal (such as, e.g., staying silent and presenting no evidence at all) is available independently of the private information held by the sender. In this case, the persuasion problem becomes tractable.

If the receiver has commitment power, she commits to taking the action if and only if being faced with a specific set of evidence. We refer to this situation as \emph{constrained delegation}, since we assume that the receiver with commitment power delegates inspection of the state of nature to a sender, whose incentive becomes to provide convincing evidence to support taking the action. Constrained delegation better fits the second example, where the company (receiver) can give the candidate (sender) a test to present evidence on the private information about qualifications, and commit to hiring her if she performs well. We show that the receiver's task in delegation is also intractable -- unless $\classP = \classNP$, optimal delegation can become hard to approximate within a factor of $2-\varepsilon$, for any constant $\varepsilon > 0$. Notably, this result applies even in instances with the natural condition of a global signal.

These computational differences complement conceptual differences known from the economics literature. Namely, persuasion lacks a condition termed ``credibility'' that was shown for delegation. Formally, credibility implies that there is a deterministic optimal solution that does not require randomization, see~\citet{GlazerR06} for details. We proceed to study algorithms with matching approximation guarantees for constrained persuasion and delegation, as well as a number of exact and approximation algorithms for various special cases. This includes, in particular, an approximation algorithm for a class of delegation problems that solves and rounds a semidefinite program (SDP). This last result might be of independent interest and, to the best of our knowledge, it is the first natural problem in information structure design, as well as mechanism design, where the SDP toolbox is used.


\section{Preliminaries}

Following~\cite{GlazerR06, Sher11, Sher14}, we study the fundamental problem of persuasion with evidence. There are two players, a sender and a receiver. The receiver is tasked with either taking a specific action and ``accept'' (henceforth $A$), or sticking to the status quo and ``reject'' (henceforth $R$). The sender wants to convince the receiver to take action $A$. There is a state of nature $\theta$ drawn from a distribution $\calD$ with support $\Theta$ of size $n$. We denote the probability that $\theta$ is drawn by $q_\theta$. 
The set $\Theta$ is partitioned into the set of acceptable states $\Theta_A$ and the set of rejectable ones $\Theta_R = \Theta \setminus \Theta_A$. We denote the total probability on acceptable states by $q_A = \sum_{\theta \in \Theta_A} q_{\theta}$, and the total probability on rejectable states by $q_R = \sum_{\theta \in \Theta_R} q_\theta$.

Both players know $\calD$. The sender knows the realization of the state of nature, the receiver does not. The sender has utility 1 whenever the receiver takes action $A$, and 0 otherwise. Formally, for the sender utility we have $u_s(A,\theta) = 1$ and $u_s(R,\theta) = 0$, for all $\theta \in \Theta$.
The utility of the receiver depends on the combination of the chosen action $a \in \{ A , R \}$ and the state of nature $\theta$. She has utility $1$ if she makes the ``right'' decision --- accept in an acceptable state of nature or reject in a rejectable state of nature --- and $0$ otherwise. Formally, $u_r(a,\theta) = 1$ when (1) $a = A$ and $\theta \in \Theta_A$, or (2) $a = R$ and $\theta \in \Theta_R$. Otherwise, $u_r(a,\theta) = 0$. 

The sender strives to send a message to the receiver according to a signaling strategy that is known to all parties. This message should persuade the receiver to accept. On the other hand, upon receiving the message, the receiver strives to infer the state of nature and make the right accept/reject decision. We focus on games with evidence, where the messages that can be sent are not arbitrary. Every state of nature has intrinsic characteristics (e.g., a candidate for a position has grades, degrees, or test scores) that \emph{can} be (but don't \emph{have} to be) revealed to the receiver, and cannot be forged.
More formally, there is a set $\Sigma$ of $m$ possible messages or \emph{signals} that the sender can report to the receiver. We are given as input a bipartite graph $H = (\Theta \cup \Sigma, E)$, where an edge $e = (\theta,\sigma) \in E$ implies that signal $\sigma$ is allowed to be sent in state $\theta$. We use $N(\theta) \subseteq \Sigma$ to denote the neighborhood of $\theta$, i.e., the set of allowed signals for state $\theta$. Similarly, $N(\sigma) \subseteq \Theta$ is the set of states in which signal $\sigma$ can be sent. To avoid trivialities, we assume that none of the neighborhoods $N(\cdot)$ are empty, i.e., there are no isolated nodes in $H$.

We study the computational complexity of games with evidence for different forms of interaction between the sender and the receiver. In the case of \emph{constrained persuasion}, the game starts with the sender committing to a \emph{signaling scheme}. A signaling scheme $\varphi$ is a mapping $\varphi : E \to [0,1]$, where $\varphi(\theta,\sigma)$ is the joint probability that state $\theta$ is realized and signal $\sigma$ is sent in state $\theta$. Clearly, for any signaling scheme we have $\sum_{\sigma \in N(\theta)} \varphi(\theta,\sigma) = q_{\theta}$ for every state $\theta \in \Theta$. After the sender has committed to a scheme $\varphi$, nature draws $\theta \in \Theta$ with probability $q_\theta$, and $\theta$ is revealed to the sender. Then, the sender sends signal $\sigma$ with probability $\varphi(\theta,\sigma)/q_{\theta}$. The receiver then decides on an action A or R. Finally, depending on the (state of nature, action)-pair, the sender and receiver get payoffs as described by the utilities above.

\begin{problem}{\textsc{(Constrained Persuasion)} $\;$}
Find a signaling scheme $\varphi^*$ for commitment of the sender such that, upon a best response of the receiver, the sender utility is as high as possible.
\end{problem}

In the case of \emph{constrained delegation}, the game starts with the receiver committing to an action for every possible signal $\sigma \in \Sigma$, according to a \emph{decision scheme}. A decision scheme $\psi$ is a mapping $\psi : \Sigma \to [0,1]$, where $\psi(\sigma)$ is the probability to choose action A. After the receiver has committed to a scheme $\psi$, nature draws $\theta \in \Theta$ with probability $q_\theta$, and $\theta$ is revealed to the sender. Then, the sender decides which signal $\sigma$ she will report (under the constraint that $\sigma \in N(\theta)$). The receiver then takes action A with probability $\psi(\sigma)$, and R otherwise. Finally, depending on the (state of nature, action)-pair, the sender and receiver get payoffs as described by the utilities above.

\begin{problem}{\textsc{(Constrained Delegation)} $\;$}
Find a decision scheme $\psi^*$ for commitment of the receiver such that, upon a best response of the sender, the receiver utility is as high as possible.
\end{problem}

Finally, in the game without commitment power, we look for a pair $(\varphi,\psi)$ of signaling and decision schemes that constitute a \emph{Bayes-Nash equilibrium} in the extensive-form game, where nature first determines the state of nature, the sender then picks $\varphi$ to provide evidence, and then the receiver uses $\psi$ to accept or reject based on the evidence provided. Given that the sender picks $\varphi$, the receiver  picks $\psi$ as a best response for every given evidence. Similarly, given that the receiver responds to evidence with $\psi$, the signaling scheme $\varphi$ is a best response for the sender. 

\begin{problem}{\textsc{(Constrained Equilibrium)} $\;$}
Find a pair of signaling scheme $\varphi$ and decision scheme $\psi$ that represents a Bayes-Nash equilibrium in the persuasion game with evidence and without commitment power.
\end{problem}

\subsection{Structural Properties}

While the persuasion problem with evidence appears rather elementary, it turns out that both persuasion and delegation variants are $\classNP$-hard, and even \classNP-hard to approximate. Hence, even in this seemingly simple domain, it is necessary to identify additional structure to obtain positive results. We mostly consider structural properties of the neighborhoods of the states of nature. 

\noindent \textbf{Unique Accepts and Rejects.} In an instance with \emph{unique accepts}, there is a single acceptable state, i.e., $|\Theta_A| = 1$. Similarly, for \emph{unique rejects} we have $|\Theta_R| = 1$. This is equivalent to assuming that every acceptable (rejectable, resp.) state $\theta$ has the same neighborhood $N(\theta)$.

\noindent \textbf{Degree-bounded States.} In an instance with \emph{degree-$k$ states}, every state $\theta \in \Theta$ has $|N(\theta)| \le k$. Similarly, for \emph{degree-$k$ accepts}, every acceptable state $\theta \in \Theta_A$ has $|N(\theta)| \le k$, and for \emph{degree-$k$ rejects} every rejectable state $\theta \in \Theta_R$ has $|N(\theta)| \le k$.

%

\noindent \textbf{Foresight.} \citet{Sher14} considers instances with \emph{foresight} defined as follows. For an acceptable state $\theta \in \Theta_A$, a signal $\sigma$ is called \emph{minimally forgeable for $\theta$} if (1)  $\sigma \in N(\theta)$, that is, $\sigma$ is valid evidence for $\theta$, and (2) $\sigma \in N(\theta')$ implies $\sigma' \in N(\theta')$ for every other signal $\sigma' \in N(\theta)$ and every rejectable state $\theta' \in \Theta_R$. In an instance \emph{with foresight} every acceptable state has a minimally forgeable signal. Intuitively, in such a problem every acceptable state $\theta$ has a (not necessarily unique) signal that is maximally informative about $\theta$ with respect to the set of rejectable states. Foresight strictly generalizes other properties studied in previous work, e.g. normality~\cite{BullW07}. Normality requires a signal for every state (not only the acceptable ones) that satisfies the condition of minimally forgeable, and it satisfies the condition w.r.t.\ all states (not only w.r.t.\ rejectable ones). In addition, foresight is a generalization of instances with unique rejects, as well as a generalization the class of degree-1 accepts.

\noindent \textbf{Global Signal.} In an instance with \emph{global signals}, there is at least one signal $\sigma$ with $N(\sigma) = \Theta$, i.e., the signal can be sent in every possible state. For example, one can think of ``being silent'' as such a global signal.

\noindent \textbf{Proof of Membership.} In an instance with \emph{proof of membership}, the set of signals $\Sigma$ is the set of all subsets of $\Theta$, and the sender is constrained so that when the state is $\theta$ she can only send a signal $\sigma$ if $\theta \in \sigma$. This special structure is also considered by~\citet{grossman1981informational} and~\citet{milgrom1981good}.
Note that this class is a special case of instances with global signal.

\noindent \textbf{Laminar Neighborhoods.} In an instance with \emph{laminar signals}, the family of neighborhoods of states $\{ N(\theta) \mid \theta \in \Theta\}$ forms a laminar family, i.e., for two states $\theta, \theta'$ the sets of allowed signals fulfill either $N(\theta) \subseteq N(\theta')$ or $N(\theta) \cap N(\theta') = \emptyset$. In an instance with \emph{laminar states}, the family of neighborhoods of signals $\{ N(\sigma) \mid \sigma \in \Sigma \}$ forms a laminar family.

In an instance with \emph{laminar states}, consider a connected component $C$ of the state-signal graph $H$. If $H$ has several connected components, the instance can be treated separately for each connected component. Let us consider a single component, or, equivalently, assume $H$ is connected. Due to laminarity, there is at least one signal $\sigma$ that has a maximal set of neighboring states, i.e., for every signal $\sigma'$ we have $N(\sigma') \subseteq N(\sigma)$. We assume that every state has an incident signal, so $N(\sigma) = \Theta$, i.e., every instance with (connected $H$ and) laminar states has global signals.

%


\subsection{Results and Contribution}

We provide polynomial-time exact and approximation algorithms as well as hardness results for the general problems and the domains with more structure described above. 

For the constrained equilibrium problem, we show that a Bayes-Nash equilibrium can always be computed in polynomial time by repeatedly solving a maximum flow problem. We compare the utility obtained in an equilibrium with the one achievable with commitment power, for the sender and the receiver, respectively. Formally, we define and bound the ratio of the utilities for best and worst-case equilibria, in the spirit of prices of anarchy and stability. For the receiver, it is known that the price of stability is $1$~\cite{GlazerR06}; we show that the price of anarchy is $2$. For the sender we show that both ratios are unbounded.
This substantial utility gain provides further motivation to study problems with commitment power.
Our results for constrained delegation and persuasion are summarized in Table~\ref{Table1}. 

\begin{table}[t]
\begin{center}
\begin{tabular}{|c||c|c||c|c|} \hline
Scenario & \multicolumn{2}{c||}{Constrained Delegation} & \multicolumn{2}{c|}{Constrained Persuasion}  \\ 
  & Upper & Lower & Upper & Lower \\ \hline\hline
General       & 2 & $2-\varepsilon$ ($\classP \ne \classNP$) & $O(n)$ & $n^{1-\varepsilon}$ ($\classP \ne \classNP$)  \\ \hline
Degree-2 States & 1.1 & \classAPX-hard~\cite{Sher14} & $O(n)$ & $n^{1-\varepsilon}$ ($\classP \ne \classNP$)  \\ \hline
Degree-$d$ States & $2-1/d^2$ & \classAPX-hard~\cite{Sher14} & $O(n)$ & $n^{1-\varepsilon}$ ($\classP \ne \classNP$)  \\ \hline
Degree-1 Rejects & 2 & \classAPX-hard~\cite{Sher14} & $O(n)$ & $n^{1-\varepsilon}$ ($\classP \ne \classNP$)  \\ \hline
Degree-1 Accepts & 1~\cite{Sher14} & & $O(n)$ & $n^{1-\varepsilon}$ ($\classP \ne \classNP$)   \\ \hline
Foresight & 1~\cite{Sher14} & & $O(n)$ & $n^{1-\varepsilon}$ ($\classP \ne \classNP$)  \\ \hline
Unique Rejects & 1~\cite{Sher14} & & 1 &  \\ \hline
Unique Accepts & 1 & & PTAS & Strongly \classNP-hard  \\ \hline
Global Signal & 2 & $2-\varepsilon$ ($\classP \ne \classNP$) & 1 &   \\ \hline
Proof of Membership & 1 &  & 1 & \\ \hline
Laminar States & 1 & & 1 & \\ \hline
Laminar Signals & 1 &  & 
& Weakly \classNP-hard \\ \hline
\end{tabular}
\caption{\label{Table1} Approximation results shown in this paper, as well as results shown or implied by~\cite{Sher14}.}
\end{center}
\end{table}

For the constrained delegation problem, we show two interesting non-trivial approximation results in Section~\ref{subsec: approximations for delegation}. For instances with degree-$d$ states we give a $(2-\frac{1}{d^2})$-approximation algorithm via LP rounding. For degree-2 states, we propose a SDP-based algorithm to compute a 1.1-approximation. To the best our knowledge, this is the first application of advanced results from the SDP toolbox in the context of information design, as well as mechanism design. 

We discuss tractable special cases in Section~\ref{subsec: efficient cases in delegation}.~\citet[Theorem 7]{Sher14} shows that in instances with foresight the optimal decision scheme can be found in polynomial time by solving a network flow problem. Unique rejects and degree-1 accepts are special cases, so the same result holds. For proof of membership, the optimal decision scheme is very simple. In case of laminar states or laminar signals, the instance does not necessarily fulfill the conditions of foresight. In both cases, we provide new polynomial-time algorithms based on dynamic programming.

For constrained persuasion, the strong hardness arises from deciding which action should be preferred by the receiver for each signal. It holds even in several seemingly special cases with degree-2 states and degree-1 accepts or degree-1 rejects. As a consequence, good approximation algorithms can be obtained only in significantly more limited scenarios than for delegation. For unique accepts, we prove strong \classNP-hardness (i.e. there is no FPTAS unless \classP = \classNP) and provide a polynomial-time approximation scheme (PTAS). In contrast, for unique rejects, the problem can even be solved in polynomial time. 

In the natural scenario when a global signal is available, we show a transformation into a standard Bayesian persuasion problem with direct signals, in which a sender with commitment power simply transmits the recommended action the receiver should take. This problem can be solved optimally in polynomial time via linear programming. This stands in strong contrast to delegation, where availability of a  global signal has no effect on the hardness of approximation. More generally, we prove the positive result for a very general version of the 2-action constrained persuasion problem with arbitrary utilities for receiver and sender. By applying the result to each component of the state-signal graph $H$, we can also obtain an optimal signaling scheme for instances with laminar states in polynomial time. The optimal signaling scheme can be obtained easily for proof of membership. Finally, for laminar signals, we show weak \classNP-hardness. It is an interesting open problem to strengthen this lower bound and to obtain a non-trivial approximation algorithm for this case.

While our hardness results hold for more general scenarios, the majority of our positive results crucially use the fact that the receiver's action is binary (accept or reject). In constrained delegation, the general $2$-approximation algorithm which picks the better of ``always accept'' and ``always reject'' can be naturally extended to the problem with more receiver actions (with the approximation guarantee degrading as the number of actions grows), but there is no natural extension for our  more specialized algorithms that beat the factor of $2$ in special cases. In constrained persuasion, a signaling scheme partitions the signal space into two sets, $\Sigma_A$ and $\Sigma_R$, in the sense that the receiver takes action $A$ if and only if she gets signal $\sigma \in \Sigma_A$ (and $R$ for $\Sigma_R$), and our positive results in crucially rely on the fact that a signal is either in a set or its complement. Therefore, to get positive results for more general settings, a new approach seems necessary. Understanding the landscape of constrained persuasion and constrained delegation in more general settings, e.g., when there are multiple receiver actions, is left as an interesting research direction.

\subsection{Related Work}\label{sec: related work}

There is a large body of literature on strategic communication, see~\citet{sobel2013giving} for an extensive review. The works most closely related to ours are~\citet{GlazerR06} and~\citet{Sher14}.~\citet{GlazerR06} introduce the problem of constrained delegation. They show, among other things, that the optimal decision scheme in constrained delegation is deterministic. Furthermore, they prove that there is always a Bayes-Nash equilibrium where the receiver plays the optimal decision scheme from constrained delegation, i.e., the price of stability for the receiver is 1. This condition is termed ``credibility'' and it is strengthened by~\citet{Sher14} to sequential equilibria. It is easy to see that this is not true when sender moves first. This conceptual difference between persuasion and delegation is reflected as a difference in the problems' computational complexity. Deterministic optimal strategies and ``credibility'' hold also beyond the simple model with 2 actions -- when receiver utility is a concave transformation of sender utility, see~\citet{Sher11}.~\citet{Sher14} builds on the model of~\citet{GlazerR06} and characterizes optimal rules for static as well as dynamic persuasion. Furthermore, and more relevant to our interest here, he proves an $\classNP$-hardness result for constrained delegation, as well as provides a polynomial-time algorithm for optimal delegation in instances with foresight. Here we strengthen this hardness result to a hardness of approximation within a factor of $2-\varepsilon$ (and provide a matching, alas trivial, approximation algorithm). While this subsumes \classNP-hardness in general, we observe that his hardness proof applies in case of degree-2 states and degree-1 rejects, and that it even implies \classAPX-hardness for such instances.

\citet{GlazerR04} study a related setting, where the state of nature is multi-dimensional, and the receiver can verify at most one dimension. The authors characterize the optimal mechanism as a solution to a particular linear programming problem, show that it takes a fairly simple form, and show that random mechanisms may be necessary to achieve the optimum.
\citet{carroll2019strategic} study the problem of fully revealing the sender's information in a setting with multidimensional states, where the receiver can verify a single dimension. Importantly, the dimension the receiver chooses to reveal depends on the sender's message.

A number of works in the algorithmic economics literature investigate the computational complexity of persuasion and information design. Computational aspects of the Bayesian persuasion model of~\citet{kamenica2011bayesian} are studied in, e.g.,~\cite{dughmi2017algorithmic, dughmi2019hardness, dughmi2016algorithmic, dughmi2019persuasion, emek2014signaling, Hahn2020secretary, Hahn2020prophet}, but in these works there are no limits on the senders' signals, i.e., $H$ is the complete bipartite graph. Closer to our work are~\citet{Dughmi2016limited} and~\citet{Gradwohl2020limited} who study computational problems in Bayesian persuasion with limited signals, where the number of signals is smaller than the number of actions.

To be consistent with most works in algorithmic economics we use the terms ``price of anarchy'' and ``price of stability'' to refer to the ratios of the optimal utility of a player with commitment power, over their utility in the worst/best equilibrium. The ``value of commitment'', the ratio of the utility of a player when she has commitment power over her utility when she does not have commitment power, is a related notion studied in~\cite{vardy2004value,letchford2014value} in the context of Stackelberg games.

\section{Equilibria, the Price of Anarchy, and the Price of Stability}

We first study the scenario without commitment power. Our interest here is to obtain a signaling scheme $\varphi : E \to [0,1]$ and a decision scheme $\psi : \Sigma \to [0,1]$, such that the pair $(\varphi,\psi)$ forms a Bayes-Nash equilibrium. 

We prove this result for a general class of games, in which the receiver has two actions (denoted A and R for consistency). Moreover, sender and receiver can have utilities $u_s, u_r : \{A,R\} \times \Theta \to \RR$ that yield arbitrary positive or negative values for every (state of nature, action)-pair. 

We conjecture that our results can be strengthened to the refined concept of \emph{sequential equilibrium} (which was studied in~\cite{GlazerR06,Sher11,Sher14}) using suitable sequences of belief systems. For simplicity, we here stick to the more straightforward notion of Bayes-Nash equilibrium.

\begin{theorem}
  \label{thm:equilibrium}
  A Bayes-Nash equilibrium can be computed in polynomial time when the receiver has two actions.
\end{theorem}

The proof of the theorem is deferred to Appendix~\ref{app: BNE Computation}.

How desirable is an equilibrium for the sender and the receiver? By how much can each player benefit when he or she enjoys commitment power? Towards this end, we bound the ratios of the optimal utility achievable with commitment power over the utilities in the worst and best equilibrium. Intuitively, commitment power might be interpreted as a form of control over the game, so we use the term \emph{price of anarchy} and \emph{price of stability} to refer to the ratios, respectively. 

More formally, for the price of anarchy we bound the ratio of the optimal utility achievable with commitment over the \emph{worst} utility in any Bayes-Nash equilibrium. For the price of stability we bound the ratio of the optimal utility achievable with commitment over the \emph{best} utility in any Bayes-Nash equilibrium.

For the receiver, the optimal scheme with commitment leads to an equilibrium~\cite{GlazerR06}, so the price of stability is 1. The price of anarchy is 2 (c.f.\ Proposition~\ref{prop: trivial 2 approximation} below). For the sender, both prices of anarchy and stability are easily shown to be unbounded. 

\begin{proposition}\label{prop:priceof}
  The price of anarchy for the receiver is 2 and this is tight. The prices of anarchy and stability for the sender are unbounded.
\end{proposition}

\begin{proof}
  For the price of anarchy for the receiver, consider any equilibrium. For every signal $\sigma \in \Sigma$, the best response for the receiver is to choose the accept/reject decision that is correct with larger conditional probability. Hence, for every signal the receiver makes the right decision with probability at least 0.5. Clearly, in the optimum he can be correct with probability at most 1.
  
  For tightness, consider one acceptable state $\theta_a$ and one rejectable state $\theta_r$, both with $q_{\theta_a} = q_{\theta_r} = 0.5$. There are two signals $\sigma_1, \sigma_2$ and three edges $(\theta_a,\sigma_1)$, $(\theta_a,\sigma_2)$ and $(\theta_r,\sigma_2)$. In the optimal scheme, the receiver sets $\psi^*(\sigma_1) = 1$ and $\psi^*(\sigma_2) = 0$ which leads to a utility of 1. In the worst equilibrium, the receiver sets $\psi^*(\sigma_2) = 1$ and the sender sets $\varphi((\theta_a,\sigma_1)) = \varphi((\theta_a,\sigma_2)) = 0.5$. The decision for $\sigma_1$ does not matter. In this case, the receiver obtains a utility of 0.5. 

  For the prices of anarchy and stability for the sender, consider one acceptable state $\theta_a$ and one rejectable state $\theta_r$, with $q_{\theta_a} = 0.25$ and $q_{\theta_r} = 0.75$. There are two signals $\sigma_1, \sigma_2$. $H$ is the complete bipartite graph. An optimal scheme for a sender with commitment turns $\sigma_1$ into an accept signal, i.e., $\varphi((\theta_a,\sigma_1)) = \varphi((\theta_r,\sigma_1)) = 0.25$ and $\varphi((\theta_r,\sigma_2)) = 0.5$. This yields a utility of 0.5 for the sender.
 
  Consider any equilibrium. A positive acceptance probability $\psi(\sigma) > 0$ requires that for signal $\sigma$ the conditional probability for $\theta_a$ is at least as high as for $\theta_r$, i.e., $\varphi((\theta_a,\sigma)) \ge \varphi((\theta_r,\sigma))$. Since $q_{\theta_a} < q_{\theta_r}$ this can happen for at most one signal. Suppose w.l.o.g.\ that this signal is $\sigma_1$, i.e., $\psi(\sigma_1) > 0$ and $\psi(\sigma_2) = 0$. Then $\varphi((\theta_r,\sigma_1)) \le 0.25$ and, thus, $\varphi((\theta_r,\sigma_2)) > 0$. Not that this implies a contradiction to the sender playing a best response against $\psi$ -- given $\psi$, it would be a better to set $\varphi((\theta_r,\sigma_1)) = 0.75$ and signal $\sigma_1$ always. This shows that in every equilibrium we have $\psi(\sigma_1) = \psi(\sigma_2) = 0$. Hence, the receiver always rejects and the sender has utility 0. Both prices in this example would be 0.5 divided by 0, i.e., unbounded. 
\end{proof}

\section{Constrained Delegation}\label{sec: delegation}

In constrained delegation, the game starts with the receiver committing to a decision scheme $\psi : \Sigma \to [0,1]$, where $\psi(\sigma)$ is the probability to choose action A if the sender reports signal $\sigma$. The first insight is due to~\citet[Proposition 1]{GlazerR06}; for completeness we include a proof in Appendix~\ref{app: proof of GR 06}.

\begin{lemma}[\citet{GlazerR06}]\label{lem: GR 06 deterministic psi}
   In constrained delegation, there is an optimal decision scheme $\psi^*$ that is deterministic, i.e., $\psi^*(\sigma) \in \{ 0 ,1 \}$ for all $\sigma \in \Sigma$.
\end{lemma}

Given a deterministic decision scheme $\psi$, the sender's problem is trivial: after learning $\theta$, report an arbitrary signal $\sigma \in N(\theta)$ such that $\psi(\sigma) = 1$ if one exists. Otherwise, report an arbitrary signal $\sigma \in N(\theta)$. In the following, we focus on the computational complexity of the receiver's problem: How hard is it to compute the optimal $\psi$? What about a good approximation algorithm?

This problem turns out to be much easier than the sender's problem in constrained persuasion studied below. It readily admits a trivial $2$-approximation algorithm. Let $\psi_A$ be the scheme that accepts all signals, i.e., $\psi_A(\sigma) = 1$ for all $\sigma$, and $\psi_R$ the scheme that rejects all signals. The better of $\psi_A$ and $\psi_R$ results in utility $\max\{ q_A , q_R \}$ for the receiver, which is, of course, at least $1/2$. Trivially, the receiver can obtain at most a utility of 1.

\begin{proposition}\label{prop: trivial 2 approximation}
For the constrained delegation problem, the better of $\psi_A$ and $\psi_R$ is a $2$-approximation to the optimal decision scheme $\psi^*$.
\end{proposition}

In Section~\ref{subsec: hardness of delegation} we show that the factor 2 is essentially optimal in the worst case, unless \classP\ = \classNP.
In Section~\ref{subsec: approximations for delegation} we present our results on approximation algorithms. The results on special cases with optimal schemes are discussed in Section~\ref{subsec: efficient cases in delegation}.

\subsection{Hardness}\label{subsec: hardness of delegation}

\citet[Theorem 7]{Sher14} shows \classNP-hardness of constrained delegation, even in the special case with degree-1 rejects. His proof easily extends to show \classAPX-hardness, even for degree-1 rejects and degree-2 states; we provide the arguments in Appendix~\ref{app:APX} for completeness. Our main result in this section is a stronger hardness result that matches the guarantee of the trivial algorithm in Proposition~\ref{prop: trivial 2 approximation}.

\begin{theorem} \label{thm:2-inapprox-dsv}
For any constant $\eps \in (0, 1)$, it is \classNP-hard to approximate constrained delegation within a factor of $(2-\eps)$.
\end{theorem}

For simplicity, we sketch below an outline for a reduction that does \emph{not} give the \classNP-hardness, but nonetheless encapsulates the main ideas of the proof. After the outline, we roughly explain the changes needed to achieve the \classNP-hardness; the full proof is deferred to Appendix~\ref{app: lower bound for delegation}. 

We reduce from the \BVE\ problem. In this problem, we are given a bipartite graph $(U, V, E)$ and positive real number $\beta$. The goal is to select (at least) $\beta |U|$ vertices from $U$ such that their neighborhood (in $V$) is as small as possible.~\citet{KhotS16} show the following strong inapproximability result:

\begin{theorem}[\cite{KhotS16}] \label{thm:bvx}
Assuming $\classNP \nsubseteq \bigcap_{\delta > 0} \DTIME(2^{n^{\delta}})$, for any positive constants $\tau, \gamma > 0$, there exists $\beta \in (0, 1)$ such that no polynomial-time algorithm can, given a bipartite graph $(U, V, E)$, distinguish between the following two cases:
\begin{itemize}
\item (YES) There exists $S^* \subseteq U$ of size at least $\beta |U|$ where $|N(S^*)| \leq \gamma |V|$.
\item (NO) For every $S \subseteq U$ of size at least $\tau \beta |U|$, $|N(S)| > (1 - \gamma) |V|$.
\end{itemize}
\end{theorem}


The main idea of our reduction is quite simple. Roughly speaking, given a bipartite graph $(U, V, E)$, we set $\Sigma = U$, $\Theta_R = V$ and the edge set between them is exactly $E$. To get a high utility on $\Theta_R$, we must pick a signal set $T \subseteq \Sigma$ such that $|N(T)|$ is small, and set $\psi(\sigma) = 1$ for all $\sigma \in T$; this does not mean much so far, since we could just pick $T = \emptyset$. This is where the set of acceptable states comes in: we let $\Theta_A$ be equal to $U^{\ell} = \{ (u_1,\dots,u_{\ell}) | u_i \in U \}$ for some appropriate $\ell \in \NN$, and there is an edge between $\theta = (u_1,\dots,u_{\ell})$ and $\sigma = u$ if $u_i = u$ for some $i \in [\ell]$. Intuitively, this forces us to pick $T$ that is not too small as otherwise $\Theta_A$ won't contribute to the total utility. Finally, we need to pick a distribution $\calD$ over $\Theta$ such that $q_A = q_R$, as otherwise the trivial algorithm already gets better than a 2-approximation. 

As stated earlier, the above reduction does not yet give \classNP-hardness, because Theorem~\ref{thm:bvx} relies on a stronger assumption\footnote{We remark that it is entirely possible that Theorem~\ref{thm:bvx} holds under \classNP-hardness (instead of under the assumption $\classNP \nsubseteq \bigcap_{\delta > 0} \DTIME(2^{n^{\delta}})$) but this is not yet known.} that $\classNP \nsubseteq \bigcap_{\delta > 0} \DTIME(2^{n^{\delta}})$. To overcome this, we instead use a ``colored version'' of the problem, where every vertex in $U$ is colored and the subset $S \subseteq U$ must only contain vertices of different colors (i.e., be ``colorful''). It turns out that the above reduction can be easily adapted to work with such a variant as well, by changing the acceptable states $\Theta_A$ to ``test'' this condition instead of the condition that $|S|$ is small. Furthermore, we show, via a reduction from the Label Cover problem, that this colored version of \BVE\ is \classNP-hard to approximate. Together, these imply Theorem~\ref{thm:2-inapprox-dsv}. Our proof formalizes this outline; see Appendix~\ref{app: lower bound for delegation} for details.
 
\paragraph{Global Signals.} In constrained delegation the existence of global signals, i.e., a set of signals that every state has access to, does not substantially change the receiver's problem (c.f.~\cite[pg.\ 103]{Sher14}). Specifically, if some global signal $\sigma$ is accepted, then $\sigma$ will be sent from every single state of nature, resulting in a trivial solution with receiver utility $q_A$. If all global signals are rejected, the receiver is left to solve the problem on the remaining, possibly arbitrary state-signal graph $H$.

\begin{corollary}
For any constant $\eps \in (0, 1)$, it is \classNP-hard to approximate constrained delegation with global signals within a factor of $(2-\eps)$.
\end{corollary}

\subsection{Approximation Algorithms}\label{subsec: approximations for delegation}

By Theorem~\ref{thm:2-inapprox-dsv} there is no hope for a $(2-\epsilon)$-approximation algorithm for the constrained delegation problem. Proposition~\ref{prop: trivial 2 approximation} provides a matching guarantee. 

As a consequence, we examine in which way instance parameters influence the existence of polynomial-time approximation algorithms. In particular, the maximum degree $d$ is a main force that drives the hardness result. For the case of degree at most $d$, we give a $2-\frac{1}{d^2}$ approximation algorithm via LP rounding. When $d = 2$, we improve upon this by giving a 1.1-approximation algorithm via SDP rounding.

\subsubsection{Better than $2$ via LP Rounding} 

For instances with degree-$d$-states we take the better of (1) rounding the natural linear program for constrained delegation and (2) the trivial scheme of Proposition~\ref{prop: trivial 2 approximation}.

\begin{theorem}\label{lem: delegation lp rounding}
For constrained delegation with degree-$d$ states there is a polynomial-time $\left(2-\frac{1}{d^2}\right)$-approximation algorithm.
\end{theorem}

\begin{proof}
Consider the following integer program for constrained delegation (c.f.\ \cite{GlazerR06,Sher14}).
\begin{maxi!}{}{\sum_{\theta \in \Theta} c_\theta q_\theta}{}{}
   \addConstraint{\sum_{\sigma \in N(\theta)} \psi_\sigma}{\geq c_\theta, \hspace{0.5cm} \text{for all } \theta \in \Theta_A \label{eq: good states cost}}
   \addConstraint{\sum_{\sigma \in N(\theta)} \psi_\sigma}{\leq | N(\theta) | (1-c_\theta) \hspace{0.5cm} \text{for all } \theta \in \Theta_R \label{eq: bad states cost}}
   \addConstraint{\psi_\sigma \in \{ 0, 1 \}, \text{for all } \sigma \in \Sigma~~ \text{ and } ~~ c_\theta \in \{ 0 , 1 \}, \text{for all } \theta \in \Theta}{}
\end{maxi!}
The variable $\psi_\sigma$ encodes whether the action is accept or reject for signal $\sigma$. The variable $c_\theta$ encodes whether the receiver makes the correct choice when the state of nature is $\theta$. Constraint~\eqref{eq: good states cost} states that, if $\theta \in \Theta_A$, she can't make the correct choice when she rejects all signals available from $\theta$. Constraint~\eqref{eq: bad states cost} states that, if $\theta \in \Theta_R$, making the correct choice means rejecting all signals available from $\theta$; the $|N(\theta)|$ term ensures that the constraint can still be satisfied even when $c_\theta = 0$.

Our algorithm first solves the linear relaxation of this integer program; let $\hat{\psi}_\sigma$ and $\hat{c}_\theta$ be the fractional optimum. We round this solution by (independently) setting $\psi_\sigma = 1$ with probability $\hat{\psi}_\sigma$, and 0 otherwise. We can optimally pick $c_\theta$ given the $\psi_\sigma$'s. The rounded solution is feasible by definition; we show that it is a good approximation to the optimal LP value, i.e., $\sum_{\theta \in \Theta} \hat{c}_\theta q_\theta$.

Let $G = \frac{1}{|\Theta_A|} \sum_{\theta \in \Theta_A} \hat{c}_\theta q_\theta$ and $B = \frac{1}{|\Theta_R|}\sum_{\theta \in \Theta_R} \hat{c}_\theta q_\theta$ be the average contribution to the LP objective from the acceptable and rejectable states, respectively. The LP value is $G |\Theta_A| + B |\Theta_R |$. We start by showing the following lower bound on the expected value of the rounded solution.

\begin{lemma}\label{lemma: rounded lp bound}
$\mathbb{E}[\sum_{\theta \in \Theta} c_\theta q_\theta] \geq \frac{G |\Theta_A|}{d} + q_R (1-d) + d B |\Theta_R|$.
\end{lemma}

\begin{proof}
First, consider a state $\theta \in \Theta_A$. The probability that $c_\theta = 1$ is at least the probability that we rounded one of the $\psi_\sigma$ variables to $1$, for $\sigma \in N(\theta)$, i.e., 
\begin{equation}\label{eq: c theta good}
 \Pr{c_\theta = 1} \geq \max_{\sigma \in N(\theta)} \hat{\psi}_\sigma \geq \frac{\hat{c}_{\theta}}{|N(\theta)|} \geq \frac{\hat{c}_{\theta}}{d}\enspace,
\end{equation}
where we used the fact that $\hat{c}_\theta$ satisfies Constraint~\eqref{eq: good states cost}. For a state $\theta \in \Theta_R$, the probability that $c_\theta = 1$ is exactly the probability that none of its signals were selected, which is $\prod_{\sigma \in N(\theta)} (1-\hat{\psi}_\sigma) \ge 1 - \sum_{\sigma \in N(\theta)} \hat{\psi}_\sigma$. Thus
\begin{equation}\label{eq: c theta bad}
 \Pr{c_\theta = 1} \geq 1 - \sum_{\sigma \in N(\theta)} \hat{\psi}_\sigma \geq  1 - | N(\theta) | (1-\hat{c}_\theta) \geq 1 - d + d \hat{c}_\theta\enspace,
\end{equation}
where we used the fact that $\hat{c}_\theta$ satisfies Constraint~\eqref{eq: bad states cost}.
Adding up~\eqref{eq: c theta good} and~\eqref{eq: c theta bad}, the expected value of our rounded solution is
\begin{align*}
\mathbb{E}\left[\sum_{\theta \in \Theta} c_\theta q_\theta\right] &\geq \sum_{\theta \in \Theta_A} \frac{q_\theta \hat{c}_{\theta}}{d}  + \sum_{\theta \in \Theta_R} q_\theta ( 1 - d + d \hat{c}_\theta )  \geq \frac{G |\Theta_A|}{d} + q_R (1-d) + d B |\Theta_R|. 
\end{align*}
\end{proof}

Our final algorithm, i.e., the better of the trivial scheme and the rounded LP solution, has expected value at least $\max\{ q_A, q_R, \mathbb{E}[\sum_{\theta \in \Theta} c_\theta q_\theta] \}$. We have that
\begin{align*}
\left(2d - \frac{1}{d}\right) &\max\left\{ q_A, q_R, \mathbb{E}\left[\sum_{\theta \in \Theta} c_\theta q_\theta\right] \right\} \geq \left(d - \frac{1}{d}\right) q_A + (d-1) q_R +  \mathbb{E}\left[\sum_{\theta \in \Theta} c_\theta q_\theta\right]\\
 &\overset{\text{Lemma}~\ref{lemma: rounded lp bound}}{\geq} \left(d - \frac{1}{d}\right) q_A + (d-1) q_R + \frac{G |\Theta_A|}{d} + q_R (1-d) + d B |\Theta_R| \\
 &\overset{(G |\Theta_A| \leq q_A)}{\geq} d G |\Theta_A|  + d B |\Theta_R|\enspace,
\end{align*}
which is $d$ times the value of the optimum fractional value of the LP. The theorem follows. 
\end{proof}

\subsubsection{Better than $2$ via Semidefinite Programming}\label{subsubsec: sdp}

In this subsection we give a 1.1-approximation algorithm for constrained delegation with degree-2 states, where every state of nature $\theta$ has at most two allowed signals, $\sigma_u$ and $\sigma_v$. The approach stems from an observation that the problem belongs to the class of \emph{constraint satisfaction problems (CSPs)}; we make use of the toolbox for semidefinite program (SDP) rounding in approximating CSPs (e.g.~\cite{GoemansW94,feige1995approximating,lewin2002improved}).

Consider the integer program~\eqref{eq: ip objective} for our problem below. We assume w.l.o.g. that every state has \emph{exactly} two adjacent signals; if there is a state $\theta$ with a single neighbor $\sigma$, we can add a parallel edge $(\theta,\sigma)$ in $H$ and the analysis remains valid. Note that the integer program here is \emph{not} the same as the one used in the previous subsection. An intuitive reason for the change is that the variables $c_{\theta}$ there are redundant: given $\{\psi_{\sigma}\}_{\sigma \in \Sigma}$, the values of $\{c_{\theta}\}_{\theta \in \Theta}$ are already fixed. In particular, each $c_{\theta}$ can be expressed as a degree-$d$ polynomial\footnote{Note that linear functions do not suffice to express $c_{\theta}$. In particular, if we rewrite~\eqref{eq: bad states cost} for $\theta = (\sigma_i, \sigma_j)$ as $c_{\theta} \leq 1 - \frac{\psi_{\sigma_i} + \psi_{\sigma_j}}{2}$, then it is still possible to have $c_\theta = 1/2$ when $\psi_{\sigma_i} = 1, \psi_{\sigma_j} = 0$. } in $\{\psi_{\sigma}\}_{\sigma \in N(\theta)}$, which is exactly how the integer program below is written.
\begin{maxi!}{x \in \{-1,1\}^m}{\frac{1}{4} \sum_{\theta = (\sigma_i,\sigma_j) \in \Theta_A} (3 - x_i - x_j - x_i x_j) q_{\theta} + \frac{1}{4} \sum_{\theta = (\sigma_i,\sigma_j) \in \Theta_R} (1 + x_i + x_j + x_i x_j) q_{\theta}}{}{\label{eq: ip objective}}
\end{maxi!}
In the program above $x_i = -1$ is interpreted as accepting when the signal is $\sigma_i$. One can check that $\frac{1}{4} \left( 3 - x_i - x_j - x_i x_j \right)$ is equal to $1$ iff at least one of $x_i ,x_j$ is $-1$ (and zero otherwise), i.e., a state of nature $\theta \in \Theta_A$ contributes to the objective only when at least one of its allowed signals is accepted. Similarly, $\frac{1}{4} (1 + x_i + x_j + x_i x_j)$ is equal to $1$ if and only if both $x_i$ and $x_j$ are equal to $1$.

We will solve the semidefinite relaxation of this program, and give a rounding algorithm. The SDP is the following, where we replaced $x_i$ by $w_i$, to distinguish these vector variables from the variables of our integer program above.

%
\begin{subequations}
\begin{align}
\max \quad &\frac{1}{4} \sum_{\theta = (\sigma_i,\sigma_j) \in \Theta_A} (3 - w_i \cdot w_0 - w_j \cdot w_0 - w_i \cdot w_j) q_{\theta} \nonumber \\
& \qquad + \frac{1}{4} \sum_{\theta = (\sigma_i,\sigma_j) \in \Theta_R} (1 + w_i \cdot w_0 + w_j \cdot w_0 + w_i \cdot w_j) q_{\theta} \label{eq: sdp objective}\\
\text{s.t.} \quad &w_i \cdot w_i = 1 \hspace{0.5cm} \text{for all } i \in [m]\cup \{ 0 \} \label{eq: dot product is one}\\
& w_i \cdot w_0 + w_j \cdot w_0 + w_i \cdot w_j \geq -1 
\hspace{0.5cm} \text{for all } i,j \in [m]\label{eq: triangle ineq1} \\
& -w_i \cdot w_0 + w_j \cdot w_0 - w_i \cdot w_j \geq -1 \hspace{0.5cm} \text{for all } i,j \in [m]\label{eq: triangle ineq2} \\ 
& -w_i \cdot w_0 - w_j \cdot w_0 + w_i \cdot w_j \geq -1 \hspace{0.5cm} \text{for all } i,j \in [m] \label{eq: triangle ineq3} \\
&w_i \in \mathbb{R}^{m+1} \hspace{0.5cm} \text{for all } i \in [m]\cup \{ 0 \} \notag
\end{align}
\end{subequations}

Constraint~\eqref{eq: dot product is one} is standard. Constraints~\eqref{eq: triangle ineq1}-\eqref{eq: triangle ineq3} encode the triangle inequalities, which are satisfied by every valid solution to the original program; these strengthen the relaxation a bit (see~\cite{feige1995approximating, lewin2002improved}). Let $\sdpopt$ denote the optimal value of this semidefinite program (SDP). We generally cannot find the exact solution to an SDP, but it is possible to find a feasible solution with value at least $\sdpopt - \epsilon$ in time polynomial in $1/\epsilon$ (see~\citet{alizadeh1995interior}). In our analysis we will (as is typically the case) ignore the $\epsilon$ factor as it can be made arbitrarily small given sufficient time.

It is known that the SDP written above provides the optimal approximation achievable in polynomial time for any 2-CSPs~\cite{Raghavendra08,RaghavendraS09a} including our problem, assuming the Unique Games Conjecture (UGC). However, a generic rounding algorithm from this line of work (see e.g.~\cite{RaghavendraS09a}) does not give a concrete approximation ratio. Below, we describe a specific family of rounding algorithms for which we can provide the concrete approximation ratio of 1.1.

\paragraph{Rounding Algorithm.} 
Given solution vectors $\{ w_0, w_1, \dots w_m \}$, $w_i \in \mathbb{R}^{m+1}$, for this SDP we produce a feasible solution $x_i \in \{ -1, 1\}$ (for $i \in [m]$) to the original integer program as follows.
Let $\xi_i = w_0 \cdot w_i$, and $\tilde{w}_i = \frac{w_i - \xi_i w_0}{\sqrt{1 - \xi_i^2}}$ be the part of $w_i$ orthogonal to $w_0$, normalized to a unit vector. Our rounding algorithm mostly follows the rounding procedure of~\citet{lewin2002improved}, which they call $\mathcal{THRESH}^-$. First, pick a $(m + 1)$-dimensional vector\footnote{In other words, the $i$-th dimension $r_i$ is sampled independently from a Gaussian with zero mean and variance one.} $r \sim \calN(0, 1)$ $r \in \mathbb{R}^{m+1}$. Then, set $x_i = -1$ (which corresponds to accepting signal $\sigma_i$) if and only if $\tilde{w}_i \cdot r \leq T(\xi_i)$, where $T(.)$ is a threshold function, and set $x_i = 1$ otherwise. Specifically, $T(x) = \Phi^{-1}( \frac{1-\nu(x)}{2} )$, where $\Phi^{-1}(.)$ is the inverse of the normal distribution function, and $\nu: [-1,1] \rightarrow [-1,1]$ is a function. Later in the analysis --- and this is essentially the point in which various SDP rounding methods diverge from each other, e.g. see~\cite{sjogren2009rigorous} for the different choices for MAX-2-SAT and MAX-2-AND --- we will optimize over a family of $\nu(.)$, exploiting structure in our problem, in order to improve our approximation ratio.

\paragraph{Generic Analysis.} 
We now derive a generic analysis for $\mathcal{THRESH}^-$ algorithms; note that these are similar arguments as in~\cite{lewin2002improved,austrin2007balanced}. However, in the end, we will pick a different function $\nu$ than previous works, which results in better approximation ratios for our problem.

First, notice that $\tilde{w}_i \cdot r$ is a standard $\calN(0,1)$ variable, and therefore by the choice of $T(.)$ we have that $\Pr{x_i = -1} = \frac{1-\nu(\xi_i)}{2}$, which implies that
\begin{equation}\label{eq: nu xi}
\Ex{ x_i } = \nu(\xi_i) \enspace.
\end{equation}
Now, we need to also analyze the quadratic terms. Let $\Gamma_c(\mu_1,\mu_2) = Pr[X_1 \leq t_1 \text{ and } X_2 \leq t_2]$, where $t_i = \Phi^{-1}(\frac{1-\mu_i}{2})$, and $X_1,X_2 \in \calN(0,1)$ with covariance $c$ (in other words, $\Gamma_c$ is the bivariate normal distribution function with covariance $c$, with a transformation on the input).

Let $\rho = w_i w_j$ and $\tilde{\rho} = \tilde{w}_i \tilde{w}_j = \frac{\rho - \xi_i \xi_j}{\sqrt{1-\xi_i^2}\sqrt{1-\xi_j^2}}$. Observe that the products $\tilde{w}_i \cdot r$ and $\tilde{w}_j \cdot r$ are $\calN(0,1)$ random variables with covariance $\tilde{\rho}$. Thus, the probability that $\tilde{w}_i \cdot r \leq T(\xi_i)$ and $\tilde{w}_j \cdot r \leq T(\xi_j)$ (i.e., both $x_i, x_j$ are set to $-1$) is exactly $\Gamma_{\tilde{\rho}}( \nu(\xi_i), \nu(\xi_j) )$. The probability that $x_i = x_j = 1$ is equal to $\Gamma_{\tilde{\rho}}( -\nu(\xi_i), -\nu(\xi_j) )$. Austrin~\cite[Proposition 2.1]{austrin2007balanced} shows that $\Gamma_c(-\mu_1,-\mu_2) = \Gamma_c(\mu_1,\mu_2) + \mu_1/2 + \mu_2/2$. Using this fact we can calculate the probability that $x_i = x_j$, which, in turn, gives that
\begin{equation}\label{eq: nu xi xj}
\Ex{ x_i x_j } = 4 \Gamma_{\tilde{\rho}}( \nu(\xi_i), \nu(\xi_j) ) + \nu(\xi_i) + \nu(\xi_j) - 1\enspace.
\end{equation}

With Equations~\eqref{eq: nu xi} and~\eqref{eq: nu xi xj} in hand we can calculate the expected value of our rounding algorithm (i.e., the expected value of~\eqref{eq: ip objective}) for every choice of $\nu$, and compare it against the value of the SDP in~\eqref{eq: sdp objective}. Specifically, we will aim for a term-by-term approximation. Define the following quantities:
\begin{align*}
 \ell^{OR}_{\nu}( \xi_i,\xi_j,\rho ) &= \frac{3 - \xi_i - \xi_j - \rho}{4 - 2\nu(\xi_i) - 2\nu(\xi_j) - 4 \Gamma_{\tilde{\rho}}( \nu(\xi_i), \nu(\xi_j) )} \\  
  \ell^{AND}_{\nu}( \xi_i,\xi_j,\rho ) &= \frac{1 + \xi_i + \xi_j + \rho}{2\nu(\xi_i) + 2\nu(\xi_j) + 4 \Gamma_{\tilde{\rho}}( \nu(\xi_i), \nu(\xi_j))}\enspace, 
\end{align*}
and let
\begin{align*}
\ell^{OR}(\nu) = \min_{\xi_i,\xi_j,\rho} \ell^{OR}_{\nu}( \xi_i,\xi_j,\rho) && 
\text{ and } && \ell^{AND}(\nu) = \min_{\xi_i,\xi_j,\rho} \ell^{AND}_{\nu}( \xi_i,\xi_j,\rho)\enspace,
\end{align*}
where the minimization is over all choices of $\xi_i, \xi_j, \rho \in [-1, 1]$ that satisfy the triangle inequalities (Constraints~\eqref{eq: triangle ineq1}-\eqref{eq: triangle ineq3}).
It is now straightforward to see that the term-by-term analysis implies that, for any choice of $\nu$, our approximation ratio is at most $\max\{\ell^{OR}(\nu), \ell^{AND}(\nu)\}$.


\paragraph{Choosing $\nu$ and Putting Things Together.}
We are left to choose the function $\nu$ that results in the smallest approximation ratio $\max\{\ell^{OR}(\nu), \ell^{AND}(\nu)\}$. We consider a rounding function of the form $\nu(y) = \alpha \cdot y + \beta$ for parameters $\alpha, \beta$ to be chosen. Using extensive computational effort, we found that $\alpha = 0.8825$ and $\beta = 0.0384$ perform well. Once we have a choice for $\alpha$ and $\beta$, it remains to prove the approximation ratio.

We have a computer-assisted proof showing that the approximation ratio is at most 1.1; our computer-based proof approach is similar to that of~\cite{sjogren2009rigorous}. Roughly speaking, we divide the cube $(\xi_i, \xi_j, \rho) \in [-1, 1]^3$ into a certain number of subcubes. For each subcube, we (numerically) compute an upper bound to $\max\{\ell^{OR}_{\nu}( \xi_i,\xi_j,\rho), \ell^{AND}_{\nu}( \xi_i,\xi_j,\rho)\}$. If this upper bound is already at most 1.1, then we are finished with the subcube. Otherwise, we divide it further into a certain number of subcubes. By continuing this process, we eventually manage to show that for the whole region $[-1, 1]^3$ that satisfies the triangle inequalities, the ratio must be at most 1.1, as desired. (The smallest subcube our proof considers has edge length 0.00078.) 

\paragraph{Comparison to Prior Work.} As stated earlier, our algorithm, with the exception of the choice of $\nu$, is similar to~\cite{lewin2002improved} and the follow-up works (e.g.~\cite{austrin2007balanced,sjogren2009rigorous}). However, perhaps surprisingly, we end up with a better approximation ratio than the \textsc{Max 2-AND} problem\footnote{This is the problem where we are given a set of clauses, each of which is an AND of two literals. The goal is to assign the variables as to maximize the number of satisfied clauses.}, whose approximation ratio is known to be at least 1.143 assuming the UGC~\cite{Austrin10}. To understand the difference, recall that \textsc{Max 2-AND} can be written as $\max \frac{1}{4} \sum_{(i, j, b_i, b_j)} (1 + b_i x_i + b_j x_j + b_i b_j x_i x_j)$ where $b_i, b_j \in \{\pm 1\}$ (representing whether the variable is negated in the clause). This is very similar to our problem~\eqref{eq: ip objective}, except that \textsc{Max 2-AND} has the aforementioned $b_i,b_j$-terms for negation. It turns out that this is also the cause that we can achieve better approximation ratio. Specifically, these negation terms led previous works~\cite{lewin2002improved,austrin2007balanced,sjogren2009rigorous,Austrin10} to only consider $\nu$ that is an odd function, i.e., $\nu(y) = \nu(-y)$ for all $x \in [-1, 1]$. For example, Austrin~\cite{austrin2007balanced} considers a function of the form $\nu(y) = \alpha \cdot y$. We note here that, due to the aforementioned UGC-hardness of \textsc{Max 2-AND}, we cannot hope to get an approximation ratio smaller than 1.143 using odd $\nu$. Nonetheless, since we do not have ``negation'' in our problem, we are not only restricted to odd $\nu$, allowing us to consider a more general family of the form $\nu(y) = \alpha \cdot y + \beta$ for $\beta \ne 0$. This ultimately leads to our better approximation ratio.

\subsection{Optimal Constrained Delegation in Polynomial Time}\label{subsec: efficient cases in delegation}


\subsubsection{Unique Accepts and Rejects} Let us briefly consider the cases in which we have a unique acceptable or a unique rejectable state. Constrained delegation with unique rejects is a special case of foresight, since in this case, for every acceptable state, every incident signal is minimally forgeable. Hence, an optimal scheme can be found in polynomial time~\cite{Sher14}. For unique accepts, there is a simple algorithm to compute an optimal decision scheme.

\begin{proposition}
  For constrained delegation with unique accepts there is a polynomial-time algorithm to compute an optimal decision scheme $\psi^*$.
\end{proposition}

\begin{proof}
  Since there is only one acceptable state $\theta_a$, an optimal decision scheme must turn at most one signal from the ones incident to $\theta_a$ into an accept signal (or simply reject all signals). There are at most $m+1$ such schemes that must be considered. The best one for the receiver is an optimal decision scheme for the instance. 
\end{proof}

\subsubsection{Proof of Membership} When the set of signals is the power set of $\Theta$, and $\sigma \in N(\theta)$ if and only if $\theta \in \sigma$, the receiver's problem is trivial: reject all signals, except signals corresponding to singleton sets $\{ \theta \}$, for $\theta \in \Theta$. This scheme is obviously optimal, an in fact results in expected utility equal to $1$ for the receiver.

\begin{proposition}
  For constrained delegation with proof of membership an optimal decision scheme can be found in polynomial time.
\end{proposition}

\subsubsection{Laminar States} For laminar states, we can compute the optimal decision scheme in polynomial time using dynamic programming.

\begin{theorem}
   For constrained delegation with laminar states there is a polynomial-time algorithm to compute an optimal decision scheme $\psi^*$.
\end{theorem}

\begin{proof}
   For laminar states, the neighborhoods of signals $N(\sigma)$ form a laminar family of states. For the rest of the proof, we assume that the state-signal graph $H$ is connected, since otherwise we can apply the algorithm separately to each connected component of $H$. Moreover, if two signals $\sigma, \sigma'$ have the same neighborhood $N(\sigma) = N(\sigma')$, then w.l.o.g.\ $\psi^*$ treats them similarly with $\psi^*(\sigma) = \psi^*(\sigma')$. Hence, for the rest of the proof we assume that every signal $\sigma$ has a unique neighborhood $N(\sigma)$.
   
   Since $H$ is connected and signal neighborhoods are unique, we can construct a new graph $T = (\Sigma, E_T)$ of signals with edge set $E_T$ as follows. There is a directed edge $(\sigma_h, \sigma_l)$ iff $N(\sigma_l) \subset N(\sigma_h)$ and there is no other signal $\sigma$ with $N(\sigma_l) \subset N(\sigma) \subset N(\sigma_h)$. Since the sets are laminar, the graph $T$ is a rooted tree, where the global signal $\sigma_0$ with $N(\sigma_0) = \Theta$ is the root of the tree. 
   
   For any signal $\sigma$, if the optimal decision scheme sets $\psi^*(\sigma) = 1$ and makes $\sigma$ an accept signal, then we can assume that the sender sends $\sigma$ for every state in $N(\sigma)$. As a consequence, we can assume w.l.o.g.\ that all descendants of $\sigma$ in $T$ are accept signals in $\psi^*$ as well. 
   
   We use this insight to compute an optimal scheme bottom-up in the tree $T$ rooted in $\sigma_0$. For any signal $\sigma$, we restrict attention to the subinstance $H_\sigma$ given by all signals in the subtree $T_\sigma$ rooted at $\sigma$ and the stats in $N(\sigma)$. Consider optimal scheme for $H_\sigma$. There are two options: (1)~$\sigma$ is an accept signal, and so are all signals in $T_\sigma$. (2) $\sigma$ is a reject signal. In this case, all states $\theta$ with $N(\theta) \cap H_\sigma = \{ \sigma \}$ would be rejected. For all other states $\theta$ we can assume that $\sigma$ is never sent by the sender, since every descendant (reject or accept) signal is weakly preferred by the sender. Hence, in case (2)~we can recurse and apply the optimal decision schemes for the instances given by the subtrees $T_{\sigma'}$ rooted at the child signals $\sigma'$ of $\sigma$.
   
   The recursive procedure now starts at the leaves of the tree and computes the optimal choice for the subinstances with single signals. Then, the procedure works bottom-up in the tree. For $\sigma$ it compares (1) the all-accept scheme to (2) the combination of the optimal schemes computed for the subtrees rooted at the children and a reject decision for $\sigma$. The better of these two schemes is the optimal decision scheme for subtree $T_\sigma$. The resulting algorithm computes an optimal decision scheme in the instance with laminar states. The overall running time is polynomial in the size of the instance. 
\end{proof}

\subsubsection{Laminar Signals} For laminar signals, we also give a polynomial-time algorithm to compute the optimal decision scheme using dynamic programming.

\begin{theorem}
   For constrained delegation with laminar signals there is a polynomial-time algorithm to compute an optimal decision scheme $\psi^*$.
\end{theorem}

\begin{proof}
We again rely on dynamic programming and a tree structure. Observe that there are subtle differences to the approach taken in the previous theorem. We again assume that the graph $H$ is connected, since otherwise we can apply the observations for each component separately.

Consider a tree graph $T=(V_T,E_T)$ defined as follows. Each node $v \in V_T$ corresponds to a subset $\Sigma_v \subseteq \Sigma$ of signals that represents a neighborhood $N(\theta)$ for at least one state $\theta \in \Theta$ (note there can be several states $\theta,\theta'$ with $N(\theta) = N(\theta')$). The vertices of the tree are ordered top-down w.r.t.\ the subset relation of the associated signal sets. A direct child of vertex $v$ satisfies $\Sigma_{v'} \subset \Sigma_v$ and there is no vertex $v'' \in V$ with $\Sigma_{v'} \subset \Sigma_{v''} \subset \Sigma_v$. Due to connectedness of $H$ and laminarity of signal sets, the root $v_0$ of $T$ has $\Sigma_{v_0} = \Sigma$. 

For each node $v$, we define a set of \emph{high signals} $\Sigma_v^h \subseteq \Sigma_v$ as follows. Consider the subtree $T_v$ rooted at $v$ with vertex set $V_v$. Then $\Sigma_v^l = \Sigma_v \setminus \left(\bigcup_{u \in V_v, u\neq v} \Sigma_u\right)$, i.e., the signals in $\Sigma_v^l$ are not present at any descendant of $v$ (but, due to the definition of $T$, at all ancestors of $v$). Note that $\Sigma_v^l$ can be empty. Moreover, every signal $\sigma \in \Sigma$ is a high signal for exactly one vertex. We say this vertex is the \emph{low vertex} of $\sigma$.

First consider the decision scheme with $\psi_R(\sigma) = 0$ for all $\sigma \in \Sigma$. Now suppose we change a set $\Sigma_A$ to become accept signals. For every low vertex $v$ of an accept signal, all states associated with $v$ and ancestors of $v$ now have an accept signal in their neighborhood. Consequently, we can w.l.o.g.\ assume that all high signals of ancestors of $v$ are also accept signals. Put differently, w.l.o.g.\ the set of low vertices of $\Sigma_A$ is``upward closed'' in the tree. This is the main structural property that drives our dynamic programming algorithm.

The algorithm works bottom-up in the tree. At each node $v$ we compute schemes for the high signals in the subtree $T_v$ rooted at $v$. More formally, the algorithm computes the best scheme (denoted $\psi_v^*$) for the instance given by $T_v$, in which we restrict to states with neighborhoods represented by nodes in $T_v$ and the signals that correspond to high signals of nodes in $T_v$. In addition, the algorithm maintains the best scheme (denoted $\psi_v^a$) for the instance $T_v$, which contains at least one accept signal.

First, suppose there is a high signal at node $v$. Then either (1) the high signal at $v$ is a reject signal (and, by upward closedness, all high signals in the subtree $T_v$ rooted at $v$ are reject signals), or (2) the high signal at $v$ is an accept signal. In case (1), the optimal signaling scheme on the high signals in $T_v$ is $\psi_R$. In case (2), the optimal signaling scheme results from making the high signal at $v$ an accept signal and using the optimal signaling scheme $\psi_{v'}$ for every subtree rooted at a direct child $v'$ of $v$. Due to the structural properties, the scheme computed in case (2) is $\psi_v^a$. The better of the two schemes from cases (1) and (2) is the optimal scheme $\psi_v^*$.

Second, suppose there is no high signal at node $v$. Then either (1) all high signals in the subtree $T_v$ rooted at $v$ are reject signals or (2) at least one high signal in $T_v$ is an accept signal. In case (1), the optimal signaling scheme on the high signals in $T_v$ is $\psi_R$. In case (2), we consider adding the optimal schemes $\psi_{v'}^*$ for every direct child $v'$ of $v$. Note that if none of these optimal schemes contains an accept signal, we violate the assumption of case (2). In this case, we consider the subtree $T_{v'}$ such that the utility difference between $\psi_{v'}^a$ and $\psi_{v'}^*$ is smallest and switch to $\psi_{v'}^a$ in this subtree. Due to the structural properties, it is straightforward to see that the scheme computed in case (2) is $\psi_v^a$. The better of the two schemes from cases (1) and (2) is the optimal scheme $\psi_v^*$.

The resulting algorithm computes an optimal decision scheme in the instance with laminar signals. The overall running time is polynomial in the size of the instance. 
\end{proof}

\section{Constrained Persuasion}
Let us now turn to the constrained persuasion problem. The sender first commits to a signaling scheme $\varphi$, which she then uses to transmit information to the receiver, once the state of nature is revealed. Given that the sender has commitment power and the receiver knows $\varphi$, the receiver picks action A if and only if conditioned on receiving signal $\sigma$, the expected utility of A is more than R, i.e.,
\[ 
\sum_{\theta \in N(\sigma) \cap \Theta_A} \varphi(\theta,\sigma) \ge \sum_{\theta \in N(\sigma) \cap \Theta_R} \varphi(\theta,\sigma) \]
or, equivalently, $2 \cdot \sum_{\theta \in N(\sigma) \cap \Theta_A} \varphi(\theta,\sigma) \ge \sum_{\theta \in N(\sigma)} \varphi(\theta,\sigma)$.

In this case, we say that $\sigma$ is an \emph{accept signal}, otherwise we call $\sigma$ a \emph{reject signal}. An optimal signaling scheme $\varphi^*$ maximizes the expected utility of the sender, i.e., the total probability associated with accept signals. Note that if both accepting and rejecting are optimal actions for the receiver, we assume that she breaks ties in favor of the sender (so, in our case, accept). This mild assumption is standard in economic bilevel problems (e.g., when indifferent between buying and not buying, a potential customer is usually assumed to buy) and is often without loss of generality. This way we avoid obfuscating technicalities 
in the definition of optimal schemes $\varphi^*$.

We study the computational complexity of finding $\varphi^*$ and polynomial-time approximation algorithms. In general, approximating $\varphi^*$ can be an extremely hard problem, even in the constrained persuasion problem. Our first insight in Section~\ref{subsec: signal partition} is that the main source of hardness in the problem is deciding on the optimal set of accept signals. We then provide a simple $2n$-approximation algorithm and an $n^{1-\varepsilon}$-hardness in Section~\ref{subsec: persuasion general}. The PTAS and the matching strong \classNP-hardness for instances with unique accepts as well as the efficient algorithm for unique rejects are discussed in Section~\ref{subsec: PTAS}. The section concludes with the discussion of instances with global signal or laminarity properties in Section~\ref{subsec: persuasion misc}.

\subsection{Signal Partitions}\label{subsec: signal partition}
A signaling scheme $\varphi$ partitions the signal space $\Sigma$ into $(\Sigma_A, \Sigma_R)$, in the sense that the receiver takes action $A$ if and only if she gets signal $\sigma \in \Sigma_A$ (and $R$ for $\Sigma_R$). Determining this partition of the signal set turns out to be the main source of computational hardness of finding $\varphi^*$: Given an optimal partition of the signal set, the reduced problem of computing appropriate optimal signaling probabilities is solved with a linear program. 

We prove this result in a general case of the persuasion problem, in which the receiver has an arbitrary finite set $\mathcal{A}$ of actions. Moreover, sender and receiver can have utilities $u_s, u_r : \mathcal{A} \times \Theta \to \RR$ that yield arbitrary positive or negative values for every (state of nature, action)-pair. 

\begin{proposition}
  \label{prop:givenSignals}
  Given a partition $P = (\Sigma_a)_{a \in \mathcal{A}}$ of the signal space such that the receiver's best action for a signal $\sigma \in \Sigma_a$ is action $a$, an optimal signaling scheme $\varphi_P^*$ for the general persuasion problem that (1) implements these receiver preferences and (2) maximizes the sender utility, can be computed by solving a linear program of polynomial size.
\end{proposition}

\begin{proof}
Given $P = (\Sigma_a)_{a \in \mathcal{A}}$, consider the following linear program~\eqref{eq:LPgivenSignals}.
\begin{equation}\label{eq:LPgivenSignals}
\begin{array}{lrcll}
\renewcommand{\arraystretch}{0.5}
\mbox{Max.} & \multicolumn{3}{l}{\D  \sum_{a \in \mathcal{A}} \sum_{\sigma \in \Sigma_a} \sum_{\theta \in N(\sigma)} x_{\theta,\sigma} \cdot u_s(a,\theta) }\\
\mbox{s.t.} & \D \sum_{\theta \in N(\sigma)} x_{\theta,\sigma} \cdot u_r(a,\theta) & \ge & \D \sum_{\theta \in N(\sigma)} x_{\theta,\sigma} \cdot u_r(a',\theta) & \hspace{0.5cm} \mbox{for all } a \in \mathcal{A}, \sigma \in \Sigma_a, a' \in \mathcal{A} \\
& \D \sum_{\sigma \in N(\theta)} x_{\theta,\sigma} & = & q_{\theta} & \hspace{0.5cm} \mbox{for all } \theta \in \Theta\\
& x_{\theta,\sigma}                          & \ge & 0 &\hspace{0.5cm} \mbox{for all } \sigma \in \Sigma, \theta \in N(\sigma)
\end{array}
\end{equation}

For each $\sigma \in \Sigma_a$ and every action $a' \neq a$ we must satisfy that $\Ex{u_r(a,\theta) \mid \sigma} \ge \Ex{u_r(a',\theta) \mid \sigma}$, encoded by the first constraint. The other two constraints encode the feasibility of the scheme. Subject to these constraints, the objective is to maximize the expected utility of the sender. An optimal LP-solution $x^*$ directly implies an optimal signaling scheme $\varphi_P^*(\theta,\sigma) = x^*_{\theta,\sigma}$. 
\end{proof}

\subsection{A 2n-Approximation Algorithm and Hardness}\label{subsec: persuasion general}

Going back to constrained persuasion with binary actions, we start by giving a simple $2n$-approximation algorithm. First, we give a useful benchmark for the optimal scheme.

\begin{lemma}
\label{lem:trivialUB}
An optimal signaling scheme $\varphi^*$ yields a sender utility of at most $\min \{1, 2 q_A\}$.
\end{lemma}

\begin{proof}
The upper bound of $1$ is trivial. $\varphi^*$ partitions the signal space into $(\Sigma_A, \Sigma_R)$, the accept and reject signals, respectively. The expected utility of the sender is
\[
\sum_{\sigma \in \Sigma_A} \sum_{\theta \in N(\sigma)} \varphi^*(\theta,\sigma) \le \sum_{\sigma \in \Sigma_A} \sum_{\theta \in N(\sigma) \cap \Theta_A}  2 \cdot \varphi^*(\theta,\sigma) \le 2 \sum_{\theta \in \Theta_A} q_{\theta} = 2 \cdot q_A\enspace. 
\]
\end{proof}

Our simple algorithm considers the $m$ partitions with a single accept signal $\Sigma_A = \{\sigma\}$, for every $\sigma \in \Sigma$. For each such partition, the algorithm determines an optimal scheme and then picks the best one, among all $m$ partitions. Instead of solving the LP of Proposition~\ref{prop:givenSignals}, given a proposed partition we proceed as follows. Assign as much probability mass from $\Theta_A \cap N(\sigma)$ to $\sigma$ and at most the same amount from $\Theta_R \cap N(\sigma)$ --- this ensures that $\sigma$ is an accept signal. The remaining probability mass is assigned arbitrarily to other signals. Note that if this is impossible, there is no scheme that makes $\sigma$ an accept signal. 

\begin{proposition}
  For constrained persuasion there is a $2n$-approximation algorithm that runs in polynomial time.
\end{proposition}
\begin{proof}
Suppose $\theta' \in \Theta_A$ is an acceptable state from which $\varphi^*$ assigns the largest amount to accept signals, i.e., $\theta' = \arg \max_{\theta \in \Theta_A} \sum_{\sigma \in \Sigma_A \cap N(\theta)} \varphi^*(\theta,\sigma)$. Clearly, the optimum accumulates on the accept signals at most $n$ times this probability mass from the set of acceptable states, and at most the same from rejectable states. Hence, $\sum_{\sigma \in \Sigma_A \cap N(\theta')} \varphi^*(\theta',\sigma) < q_{\theta'}$ is at least a $1/(2n)$-fraction of the optimal sender utility.

Consider the accept signals $\Sigma_A$ in $\varphi^*$ and any such signal $\sigma' \in N(\theta') \cap \Sigma_A$. When our algorithm checks the partition with $\sigma'$ as the unique accept signal, it finds a feasible scheme, since the optimum scheme makes $\sigma'$ an accept signal and the algorithm only assigns more probability from $\Theta_A$ to $\sigma'$. The value of this solution is at least $q_{\theta'}$. 
\end{proof}

In addition to this simple algorithm, we show a number of strong hardness results for constrained persuasion. The proofs of the following two theorems are relegated to Appendix~\ref{app: persuasion hardness}.
\begin{theorem}
  \label{thm:verifyHard}
  For any constant $\varepsilon > 0$, constrained persuasion is \classNP-hard to approximate within a factor of $n^{1-\varepsilon}$, even for instances with degree-2 states and degree-1 accepts.
\end{theorem}
For instances with degree-1 rejects a similar result follows with an adjustment of the reduction. 
\begin{theorem}
  \label{thm:verifyHard2}
   For any constant $\varepsilon > 0$, constrained persuasion is \classNP-hard to approximate within a factor of $n^{1-\varepsilon}$, even for instances with degree-2 states and degree-1 rejects.
\end{theorem}
In contrast to constrained delegation, the optimal signaling scheme for constrained persuasion does not necessarily have the ``credibility'' property, i.e., it might not represent a best response against the induced behavior of the receiver in the game without commitment power. As such, it is a natural question to ask for the best signaling scheme that enjoys this property: Compute the best signaling scheme of any Bayes-Nash equilibrium of the game without commitment power. We term this problem \emph{constrained persuasion with equilibrium schemes}.

Inspecting the reduction of Theorem~\ref{thm:verifyHard2}, we observe that in each of these instances the optimal signaling scheme has this property, i.e., it is a best response against the induced action scheme of the receiver. As such, constrained persuasion with equilibrium schemes is also \classNP-hard to approximate within a factor of $n^{1-\varepsilon}$, even for instances with degree-2 states and degree-1 rejects.

\begin{corollary}
	\label{cor:verifyHard3}
	For any constant $\varepsilon > 0$, constrained persuasion with equilibrium schemes is \classNP-hard to approximate within a factor of $n^{1-\varepsilon}$, even for instances with degree-2 states and degree-1 rejects.
\end{corollary}

\subsection{Unique Accepts and Rejects}
\label{subsec: PTAS}

\subsubsection{Unique Accepts}

In this section, we examine instances in which there is only a single acceptable state, for which we prove \classNP-hardness and give a PTAS. It will be convenient to state a lemma which allows us to get a better handle on the sender utility in an optimal signaling scheme for a given signal partition. This lemma will be helpful in both our hardness and algorithm analyses.

To state this lemma, we need some additional notation: for every subset $\tSigma \subseteq \Sigma$, we use $\Theta_R(\tSigma)$ to denote $\{\theta \in \Theta_R \mid N(\theta) \subseteq \tSigma\}$; when $\tSigma = \{\sigma\}$ is a singleton, we write $\Theta_R(\sigma)$ in place of $\Theta_R(\{\sigma\})$ for brevity. Moreover, let $N(\tSigma)$ denote $\bigcup_{\sigma \in \tSigma} N(\sigma)$. The lemma can now be stated as follows.

\begin{lemma} \label{lem:unique-accept-characterization}
Suppose that there exists a unique accept state $\theta_a$. For any partition $P = (\Sigma_A, \Sigma_R)$ of the signal space such that $\Sigma_A \ne \emptyset$, we have
\begin{enumerate}
\item There exists a signaling scheme $\varphi$ such that every signal in $\Sigma_A$ is accepted and every signal in $\Sigma_R$ is rejected by the receiver if and only if $\Sigma_A \subseteq N(\theta_a)$ and $\sum_{\theta \in \Theta_R(\Sigma_A)} q_{\theta} \leq q_{\theta_a}$.
\item When the above condition holds, any optimal signaling scheme $\varphi^*$ for the sender has utility equal to \[ \min\{2q_{\theta_a}, \sum_{\theta \in N(\Sigma_A)} q_{\theta}\},\] and, such a signaling scheme can be computed in polynomial time.
\end{enumerate}
\end{lemma}

We remark that the algorithm for finding $\varphi^*$ in the above lemma is a simple greedy algorithm that tries to ``put as much probability mass from rejectable states as possible'' in $\Sigma_A$ and then use the probability mass of the acceptable state $\theta_a$ to ``balance out'' the mass from the rejectable states, so that eventually the signals in $\Sigma_A$ are accepted. This is in contrast to the generic linear program-based algorithm in Proposition~\ref{prop:givenSignals}. The simpler greedy algorithm allows us to consider more concrete conditions and exactly compute the utility as stated in Lemma~\ref{lem:unique-accept-characterization}.

\begin{proof}[Proof of Lemma~\ref{lem:unique-accept-characterization}]
\begin{enumerate}
\item $(\Rightarrow)$ First, assume that there is such a signaling scheme $\varphi$. Clearly, every signal not in $N(\theta_a)$ must be rejected, which implies that $\Sigma_A \subseteq N(\theta_a)$. Furthermore, for all $\sigma \in \Sigma_A$, we must have $\varphi(\theta_a, \sigma) \geq \sum_{\theta \in N(\sigma) \cap \Theta_R} \varphi(\theta, \sigma)$. Summing up over all $\sigma \in \Sigma_A$ gives
\begin{align*}
q_{\theta_a} &\geq \sum_{\sigma \in \Sigma_A} \sum_{\theta \in N(\sigma) \cap \Theta_R} \varphi(\theta, \sigma) \\
&\geq \sum_{\sigma \in \Sigma_A} \sum_{\theta \in \Theta_R(\Sigma_A)} \varphi(\theta, \sigma)
= \sum_{\theta \in \Theta_R(\Sigma_A)} \sum_{\sigma \in \Sigma_A} \varphi(\theta, \sigma)
= \sum_{\theta \in \Theta_R(\Sigma_A)} q_{\theta}.
\end{align*}

$(\Leftarrow)$ Assume that $\emptyset \ne \Sigma_A \subseteq N(\theta_a)$ and $\sum_{\theta \in \Theta_R(\Sigma_A)} q_{\theta} \leq q_{\theta_a}$. We may construct a desired signaling scheme $\varphi$ as follows. First, we assign $\varphi(\theta, \sigma)$ arbitrarily for all $\theta \in \Theta_R(\Sigma_A)$. Then, we assign $\varphi(\theta_a, \sigma)$ such that $\varphi(\theta_a, \sigma) = 0$ for all $\sigma \notin \Sigma_A$ and that $\varphi(\theta_a, \sigma) \geq \sum_{\theta \in \Theta_R(\Sigma_A)} \varphi(\theta, \sigma)$ for all $\sigma \in \Sigma_A$. The former is possible because $\Sigma_A \ne \emptyset$ and the latter possible because $\sum_{\theta \in \Theta_R(\Sigma_A)} q_{\theta} \leq q_{\theta_a}$. Finally, for each $\theta \in \Theta_R \setminus \Theta_R(\Sigma_A)$, assign $\varphi(\theta, \sigma) = 0$ for all $\sigma \in \Sigma_A$. It is straightforward from the construction that this $\varphi$ is a desired signaling scheme.
\item First, we will show that any signaling scheme $\varphi$ has utility at most $\min\{2q_{\theta_a}, \sum_{\theta \in N(\Sigma_A)} q_{\theta}\}$ for the sender. Observe that the upper bound $2q_{\theta_a}$ follows trivially from Lemma~\ref{lem:trivialUB}. Thus, it suffices for us to prove that the utility is at most $\sum_{\theta \in N(\Sigma_A)} q_{\theta}$. To do so, let us rearrange the utility as follows:
\begin{align*}
\sum_{\sigma \in \Sigma_A} \sum_{\theta \in N(\sigma)} \varphi(\theta, \sigma) \leq \sum_{\theta \in N(\Sigma_A)} \sum_{\sigma \in N(\theta)} \varphi(\theta, \sigma)
= \sum_{\theta \in N(\Sigma_A)} q_{\theta}.
\end{align*}

Finally, we will construct a signaling scheme $\varphi^*$ with utility equal to $\min\{2q_{\theta_a}, \sum_{\theta \in N(\Sigma_A)} q_{\theta}\}$. The algorithm is a modification of the algorithm from the first part, and it works in four steps:
\begin{itemize}
\item For every $\theta \in \Theta_R(\Sigma_A)$, assign $\varphi(\theta, \sigma)$ arbitrarily.
\item For every $\theta \in (N(\Sigma_A) \cap \Theta_R) \setminus \Theta_R(\Sigma_A)$, assign $\varphi(\theta, \sigma)$ so that 
\[ \sum_{\sigma \in \Sigma_A} \sum_{\theta \in N(\sigma) \cap \Theta_R} \varphi(\theta, \sigma) = \min\{q_{\theta_a}, \sum_{\theta \in N(\Sigma_A) \cap \Theta_R} q_{\theta}\}.\] 
(Note that this step is possible because $\sum_{\theta \in \Theta_R(\Sigma_A)} q_{\theta} \leq q_{\theta_a}$.)
\item Assign $\varphi(\theta_a, \sigma)$ so that $\varphi(\theta_a, \sigma) = 0$ for all $\sigma \notin \Sigma_A$, and that $\varphi(\theta_a, \sigma) \geq \sum_{\theta \in N(\sigma) \cap \Theta_R} \varphi(\theta, \sigma)$ for all $\sigma \in \Sigma_A$. (Note that this is possible because, from the previous step, we must have $\sum_{\sigma \in \Sigma_A} \sum_{\theta \in N(\sigma) \cap \Theta_R} \varphi(\theta, \sigma) \leq q_{\theta_a}$.)
\item All other remaining assignments are made arbitrarily in order to turn $\varphi$ into a feasible signaling scheme.
\end{itemize}
It is straightforward to check that $\varphi^*$ is the desired signaling scheme with utility equal to 
\[ q_{\theta_a} + \min\{q_{\theta_a}, \sum_{\theta \in N(\Sigma_A) \cap \Theta_R} q_{\theta}\} = \min\{2q_{\theta_a}, \sum_{\theta \in N(\Sigma_A)} q_{\theta}\}.\]
\end{enumerate}
\end{proof}

With Lemma~\ref{lem:unique-accept-characterization} ready, we now prove \classNP-hardness of the problem.

\begin{theorem}
  Constrained persuasion with unique accepts is \classNP-hard.
\end{theorem}

\begin{proof}
  We reduce from the \MkVC\ problem, where we have a graph $G=(V,E)$. The goal is to choose a set $V'$ of $k$ vertices in order to maximize the number of edges incident to at least one vertex in $V'$. For every vertex $v \in V$, let $E(v)$ be the set of incident edges, then we try to pick a subset $V'$ of $k$ vertices to maximize $|\bigcup_{v \in V'} E(v)|$. 
   
  For each edge $e \in E$, we introduce a rejectable state $\theta_e$ with $q_{\theta_e} = \frac{1}{(|V|+k)(|E|+1)+2|E|}$. For each vertex $v$ we introduce a signal $\sigma_v$. The graph $H$ between states and signals expresses the incident property of edges and vertices. In addition, for each signal $\sigma$, we introduce auxiliary rejectable states that have $\sigma$ as their unique signal. Each auxiliary state $\theta$ has $q_{\theta} = \frac{|E| + 1}{(|V|+k)(|E|+1)+2|E|}$. Finally, the unique acceptable state $\theta_a$ is incident to all signals and has probability $$q_{\theta_a} = \frac{k(|E|+1) + |E|}{(|V|+k)(|E|+1)+2|E|}\enspace.$$  
  
  From Lemma~\ref{lem:unique-accept-characterization}, the optimal signaling scheme has sender utility equal to
  \begin{align*}
  \max_{\Sigma_A} \min\left\{2q_{\theta_a}, \sum_{\theta \in N(\Sigma_A)} q_{\theta}\right\},
  \end{align*}
  where the maximum is over non-empty $\Sigma_A \subseteq \Sigma$ such that $\sum_{\theta \in \Theta_R(\Sigma_A)} q_{\theta} \leq q_{\theta_a}$. Notice that, in our construction, this condition is satisfied iff $|\Sigma_A| \leq k$. This means that $\Sigma_A = \{\sigma_v\}_{v \in V'}$ for some subset $V'$ of size at most $k$. It is also not hard to see that 
  \begin{align*}
  \min\left\{2q_{\theta_a}, \sum_{\theta \in N(\Sigma_A)} q_{\theta}\right\} = \sum_{\theta \in N(\Sigma_A)} q_{\theta} = \frac{(|V'| + k)(|E|+1) + |\bigcup_{v \in V'} E(v)|}{(|V|+k)(|E|+1)+|E|}.
  \end{align*}
  In other words, the utility is maximized iff $V'$ is an optimal solution to the instance of \MkVC. Since the latter is \classNP-hard, we can conclude that constrained persuasion with unique accepts is also \classNP-hard. 
\end{proof}

We next give a PTAS for the problem. Before we formalize our PTAS, let us give an informal intuition. Observe that the condition in Lemma~\ref{lem:unique-accept-characterization} implies that $q_{\theta_a} \geq \sum_{\sigma \in \Sigma_A} \left(\sum_{\theta \in \Theta_R(\sigma)} q_{\theta}\right)$. This latter constraint is a \emph{knapsack constraint}. One generic strategy to solve knapsack problems is to first brute-force enumerate all possibilities of selecting ``heavy items'', which in our case are the signals with large $\sum_{\theta \in \Theta_R(\sigma)} q_{\theta}$. Then, use a simple greedy algorithm for the remaining ``light items''. Our PTAS follows this blueprint. However, since neither our constraints nor our objective function are exactly the same as in knapsack problems, we cannot use results from there directly and have to carefully argue the approximation guarantee ourselves.

\begin{theorem}
For constrained persuasion with unique accepts, for every fixed $\varepsilon \in (0, 1]$, Algorithm~\ref{algo: ptas for unique accepts} runs in time $m^{O(1/\varepsilon)} n^{O(1)}$ and outputs a $(1+\varepsilon)$-approximate solution.
\end{theorem}

\begin{algorithm}[t]
\caption{A PTAS for Constrained Persuasion with Unique Accepts.}\label{algo: ptas for unique accepts}
\DontPrintSemicolon
\KwIn{Graphs $H$ with a single acceptable state $\theta_a$, and $\varepsilon > 0$.}
\KwOut{Signaling scheme $\varphi_{\ALG}$.}

Set $\Sigma_{\geq \varepsilon} =  \left\{ \sigma \in \Sigma \growingmid \sum_{\theta \in \Theta_R(\sigma)} q_{\theta} \geq \varepsilon q_{\theta_a}\right\}$ and $\Sigma_{< \varepsilon} = \Sigma \setminus \Sigma_{\geq \varepsilon}$.\;
Initialize $\varphi_{\ALG}$ as an arbitrary signaling scheme.\;
\For{every (possibly empty) subset $S \subseteq \Sigma_{\geq \varepsilon}$ of size at most $1 / \varepsilon$}{\label{step:enum-large-signals}
Let $T = S$\;
\While{$\sum_{\theta \in \Theta_R(T)} q_{\theta} \leq q_{\theta_a}$}{\label{step:loop-adding-signals}
If the utility of $\varphi_{\ALG}$ is less than $\min\{2q_{\theta_a}, \sum_{\theta \in N(T)} q_{\theta}\}$, then let $\varphi_{\ALG}$ be the optimal\\
\ \ signaling scheme consistent with signaling partition $\Sigma_A = T$, which can be computed in \\
\ \ polynomial time due to Lemma~\ref{lem:unique-accept-characterization}. \label{step:update-best-signaling}\;
If $T = \Sigma_{< \varepsilon} \cap N(\theta_a)$, break from the loop.\;
\ \ Otherwise, add an arbitrary signal from $(\Sigma_{< \varepsilon} \cap N(\theta_a)) \setminus T$ to $T$.
  }
  }
\end{algorithm}

\begin{proof}
It is clear that our algorithm runs in time $m^{O(1/\varepsilon)} n^{O(1)}$. Let $\varphi^*$ be any optimal signaling scheme, with utility $\OPT$ for the sender. We prove that the utility of $\varphi_{\ALG}$ is at least $(1 - 0.5\epsilon) \OPT$. 

Without loss of generality we assume that the utility of $\varphi^*$ is non-zero. Now, let $(\Sigma^*_A, \Sigma^*_R)$ denote the signal partition of $\varphi^*$; since the utility of $\varphi^*$ is non-zero, we must have $\Sigma^*_A \ne \emptyset$. Furthermore, the first item of Lemma~\ref{lem:unique-accept-characterization} implies that $\Sigma^*_A \cap \Sigma_{\geq \varepsilon}$ must be of size at most $1 / \varepsilon$. As a result, our algorithm must consider $S = (\Sigma^*_A \cap \Sigma_{\geq \varepsilon})$ in the for-loop~\eqref{step:enum-large-signals}. For this particular $S$, let $T'$ denote the largest $T$ for which Line~\eqref{step:update-best-signaling} is executed. We next consider two cases, based on whether or not we have $T' = S \cup (\Sigma_{< \varepsilon} \cap N(\theta_a))$.

\begin{itemize}
\item Case I: $T' = S \cup (\Sigma_{< \varepsilon} \cap N(\theta_a))$. Notice that $T' \supseteq \Sigma_A^*$. Lemma~\ref{lem:unique-accept-characterization}, implies that the utility of $\varphi_{\ALG}$ must be at least $\OPT$.

\item Case II:  $T' \ne S \cup (\Sigma_{< \varepsilon} \cap N(\theta_a))$. This means that there exists a signal $\sigma^* \in (\Sigma_{< \varepsilon} \cap N(\theta_a))$ whose addition to $T'$ breaks the condition of the while-loop~\eqref{step:loop-adding-signals}, i.e.,
$q_{\theta_a} < \sum_{\theta \in \Theta_R(T' \cup \{\sigma^*\})} q_{\theta}$.
The right hand side of this inequality is equal to
\begin{align*}
\sum_{\substack{\theta \in \Theta_R \\ N(\theta) \subseteq (T' \cup \{\sigma^*\})}} q_{\theta} &\leq \sum_{\substack{\theta \in \Theta_R\\ N(\theta) \cap T' \ne \emptyset}} q_{\theta} + \sum_{\substack{\theta \in \Theta_R \\ N(\theta) = \{\sigma^*\}}} q_{\theta} \\
&= \sum_{\theta \in N(T') \cap \Theta_R} q_{\theta} + \sum_{\theta \in \Theta_R(\sigma^*)} q_\theta \\
&< \sum_{\theta \in N(T') \cap \Theta_R} q_{\theta} + \varepsilon q_{\theta_a}\enspace,
\end{align*}
where the last inequality since $\sigma$ belongs to $\Sigma_{< \varepsilon}$. Combining the two inequalities we have
\begin{align} \label{eq:lower-bound-util-intermediate}
\sum_{\theta \in N(T') \cap \Theta_R} q_{\theta} > (1 - \varepsilon)q_{\theta_a}\enspace.
\end{align}

On the other hand, from Lemma~\ref{lem:unique-accept-characterization}, when we execute line Line~\eqref{step:update-best-signaling} for $T = T'$, it must result in a signaling scheme of utility
\begin{align*}
\min\left\{2q_{\theta_a}, \sum_{\theta \in N(T')} q_{\theta} \right\} &= \min\left\{2q_{\theta_a}, q_{\theta_a} + \sum_{\theta \in N(T') \cap \Theta_R} q_{\theta} \right\} \overset{\eqref{eq:lower-bound-util-intermediate}}{>} (2 - \epsilon) q_{\theta_a}\enspace,
\end{align*}
which is at least $(1 - 0.5\varepsilon)\OPT$ due to Lemma~\ref{lem:trivialUB}.
\end{itemize}
Hence, we can conclude that our algorithm always outputs a signaling scheme with sender utility at least $(1 - 0.5\varepsilon)\OPT$. In other words, its approximation ratio is at most $\frac{1}{1 - 0.5\varepsilon} \leq 1 + \varepsilon$. 
\end{proof}

\subsubsection{Unique Rejects}
In contrast to the case with unique accepts studied in Section~\ref{subsec: PTAS} above, the problem can be solved in polynomial time for the unique rejects case. The main insight is that we can restrict attention to signaling schemes with at most $1$ reject signal, and then use Proposition~\ref{prop:givenSignals}.

\begin{lemma}
  For constrained persuasion with unique rejects there is a polynomial-time algorithm to compute the optimal signaling scheme $\varphi^*$.
\end{lemma}

\begin{proof}
  We denote the single rejectable state by $\theta_r$ and its set of incident signals by $\Sigma'$. Note that all signals in $\Sigma \setminus \Sigma'$ must be accept signals. Our scheme $\varphi^*$ sends deterministic signals for acceptable states, but possibly a randomized one for $\theta_r$. First, for every acceptable state we pick an incident signal from $\Sigma'$ if possible. Now consider two cases.
  
  If the total probability mass of acceptable states incident to $\Sigma'$ is more than $q_{\theta_r}$, 
  when the state of nature is $\theta_r$ our signaling scheme will randomize over signals in $\Sigma'$ in a way that all signals become accept signals. Consequently, all $\sigma \in \Sigma$ are accept signals. This is obviously optimal for the sender.
  
  If the total probability mass of acceptable states incident to $\Sigma'$ is less than $q_{\theta_r}$, it suffices to create a single reject signal in $\Sigma'$. Suppose $\sigma \in \Sigma'$ is chosen to become the unique reject signal. Then we can use Proposition~\ref{prop:givenSignals} to compute an optimal signaling scheme with $\Sigma_R = \{\sigma\}$ and $\Sigma_A = \Sigma \setminus \{\sigma\}$. There are at most $m$ signals in $\Sigma'$, hence, constructing an optimal scheme for each of them can be done in polynomial time. Among these $m$ schemes, the one that maximizes the sender utility is an optimal scheme $\varphi^*$. 
\end{proof}

\subsection{Global Signal and Laminarity}\label{subsec: persuasion misc}

\subsubsection{Global Signal}
In this section, we study the natural scenario with a global signal $\sigma_0 \in \Sigma$ that can be sent from every state of nature. We think of this as a ``stay silent'' or ``no evidence'' option. We consider the general persuasion problem (as in Proposition~\ref{prop:givenSignals}) with $k = 2$ actions. 

\begin{theorem}
  \label{thm:2-action-persuasion}
  For the two-action constrained persuasion problem with global signal there is a polynomial-time algorithm to compute the optimal signaling scheme $\varphi^*$.
\end{theorem}
\begin{proof}
The main idea is that the problem can be reduced to Bayesian persuasion using Proposition~\ref{prop:givenSignals}. Consider an optimal partition $(\Sigma_A, \Sigma_R)$ of the signal set into accept and reject signals, and w.l.o.g.\ assume that $\sigma_0 \in \Sigma_A$. Given a scheme $\varphi$ such that there is $\sigma_1 \in \Sigma_A$ with $\sigma_1 \neq \sigma_0$, we design another scheme $\varphi'$ that never uses $\sigma_1$: we set $\varphi'(\theta, \sigma_0) = \varphi(\theta, \sigma_0) + \varphi(\theta, \sigma_1)$ and $\varphi'(\theta, \sigma_1) = 0$ for all $\theta \in \Theta$.  Since in $\varphi'$ the signal $\sigma_1$ is never issued, w.l.o.g.\ we can assume that in this scheme $\sigma_1 \in \Sigma_R$. Moreover, in $\varphi'$ the signal $\sigma_0$ represents an accumulation of the accept signals $\sigma_1$ and $\sigma_0$ from $\varphi$, so in both schemes the receiver prefers the accept action when given $\sigma_0$. As a consequence, both $\varphi'$ and $\varphi$ yield the same expected utility for the sender. 
Therefore, by repeating this argument, we see that there is an optimal scheme with $\Sigma_A = \{\sigma_0\}$ and $\Sigma_R = \Sigma \setminus \{\sigma_0\}$.

Thus, we only need to consider two partitions, $(\{\sigma_0\}, \Sigma \setminus \{\sigma_0\})$ and $ (\Sigma \setminus \{\sigma_0\}, (\{\sigma_0\})$. For each of the partitions we solve the LP in Proposition~\ref{prop:givenSignals}. If the LP is feasible, we obtain an optimal scheme for the corresponding partition of signals. The better of the two schemes represents an optimal scheme $\varphi^*$ for the general 2-action persuasion problem with silence. 
\end{proof}

\paragraph{Proof of Membership.}
Proof of membership structure is a special case of the global signal property. Here we can limit attention to $|\Theta_R|$ signals of the form $\sigma_\theta = \{ \theta \}$, corresponding to states $\theta \in \Theta_R$, as well as the global signal corresponding to $\Theta$. If $q_A \ge q_R$, the optimal scheme $\varphi$ only sends the global signal and obtains sender utility 1. Otherwise, $\varphi$ sends the global signal for all acceptable states, and for an arbitrary portion of $q_A$ from the rejectable states. The global signal is still acceptable, and therefore the sender gets expected utility $2q_A$. By Lemma~\ref{lem:trivialUB} the scheme is optimal.

\begin{proposition}
    For constrained persuasion with proof of membership there is a polynomial-time algorithm to compute an optimal signaling scheme $\varphi^*$.
\end{proposition}

\paragraph{Laminar States}
As outlined in the model section, we can compute the optimal signaling scheme for each component of the state-signal graph $H$ separately. Due to the laminarity of signal neighborhoods, in each component there is a signal $\sigma$ that has a maximal set of incident states, which must be the all states of this component. Hence, $\sigma$ represents a global signal in this component. We can apply Theorem~\ref{thm:2-action-persuasion} to obtain the following corollary.

\begin{corollary}
    For constrained persuasion with laminar states there is a polynomial-time algorithm to compute an optimal signaling scheme $\varphi^*$.
\end{corollary}

\subsubsection{Laminar Signals}

In contrast to laminar states, the condition of laminarity for state neighborhoods does not result in a polynomial-time algorithm.

\begin{theorem}
  Constrained persuasion with laminar signals is \classNP-hard.
\end{theorem}

\begin{proof}
We reduce from the \PART\ problem. In this problem we are given $n$ positive integers $a_1,\ldots,a_n$. We denote their sum by $A = \sum_{i=1}^n a_i$. The goal is to decide whether there is a subset $S \subset \{1,\ldots,n\}$ such that $\sum_{i \in S} a_i = A/2$. 

For each integer $a_i$, we introduce a signal $\sigma_i$. For each signal there is a rejectable state $\theta_i$ with $N(\theta_i) = \{\sigma_i\}$ and probability $q_{\theta_i} = 2a_i/(3A)$. Finally, there is a global acceptable state $\theta_A$ with $N(\theta_A) = \Sigma$ and $q_{\theta_A} = A/(3A) = 1/3$. 

If the \PART\ instance has a solution $S$, then we choose accept signals $\Sigma_A = \{ \sigma_i \mid i \in S\}$. We distribute the probability mass of $\theta_A$ to exactly match the mass of $\theta_i$ for each $i \in S$. In this way, the sender obtains a utility of $2q_{\theta_A} = 2/3$, which is optimal by Lemma~\ref{lem:trivialUB}.

Conversely, suppose the sender obtains a utility of $2q_{\theta_A} = 2/3$. Then the signaling scheme must assign the probability mass of $\theta_A$ in a way such that it is exactly matched by the mass of the rejectable states for every accept signal. Consequently, the set of accept signals satisfies $\sum_{\sigma_i \in \Sigma_A} q_{\theta_i} = \sum_{\sigma_i \in \Sigma_A} 2a_i/(3A) = A/(3A)$, or put differently, $\sum_{\sigma_i \in \Sigma_A} a_i = A/2$. Hence, the set of accept signals infers a solution for the \PART\ instance. 
\end{proof}

We leave a non-trivial approximation algorithm for constrained persuasion with laminar signals as an interesting open problem.

\section*{Acknowledgments}

This work was done in part while Martin Hoefer and Alexandros Psomas were visiting the Simons Institute for the Theory of Computing. The authors acknowledge financial support and the invitation by the organizers to join the stimulating work environment. This work was done in part while Alexandros Psomas was visiting Google Research, Mountain View. 

Martin Hoefer is supported by DFG grants Ho 3831/6-1, 7-1, and 9-1.

Alexandros Psomas is supported in part by an NSF CAREER award CCF-2144208, a Google AI for Social Good award, and research awards from Google and Supra.

\bibliographystyle{ACM-Reference-Format} 
\bibliography{refs.bib}


\begin{thebibliography}{34}


\ifx \showCODEN    \undefined \def \showCODEN     #1{\unskip}     \fi
\ifx \showDOI      \undefined \def \showDOI       #1{#1}\fi
\ifx \showISBNx    \undefined \def \showISBNx     #1{\unskip}     \fi
\ifx \showISBNxiii \undefined \def \showISBNxiii  #1{\unskip}     \fi
\ifx \showISSN     \undefined \def \showISSN      #1{\unskip}     \fi
\ifx \showLCCN     \undefined \def \showLCCN      #1{\unskip}     \fi
\ifx \shownote     \undefined \def \shownote      #1{#1}          \fi
\ifx \showarticletitle \undefined \def \showarticletitle #1{#1}   \fi
\ifx \showURL      \undefined \def \showURL       {\relax}        \fi
\providecommand\bibfield[2]{#2}
\providecommand\bibinfo[2]{#2}
\providecommand\natexlab[1]{#1}
\providecommand\showeprint[2][]{arXiv:#2}

\bibitem[\protect\citeauthoryear{Alimonti and Kann}{Alimonti and Kann}{2000}]%
        {AlimontiK00}
\bibfield{author}{\bibinfo{person}{Paola Alimonti} {and} \bibinfo{person}{Viggo
  Kann}.} \bibinfo{year}{2000}\natexlab{}.
\newblock \showarticletitle{Some {APX}-completeness results for cubic graphs}.
\newblock \bibinfo{journal}{\emph{Theor.\ Comput.\ Sci.}}
  \bibinfo{volume}{237}, \bibinfo{number}{1-2} (\bibinfo{year}{2000}),
  \bibinfo{pages}{123--134}.
\newblock


\bibitem[\protect\citeauthoryear{Alizadeh}{Alizadeh}{1995}]%
        {alizadeh1995interior}
\bibfield{author}{\bibinfo{person}{Farid Alizadeh}.}
  \bibinfo{year}{1995}\natexlab{}.
\newblock \showarticletitle{Interior point methods in semidefinite programming
  with applications to combinatorial optimization}.
\newblock \bibinfo{journal}{\emph{SIAM J. Opt.}} \bibinfo{volume}{5},
  \bibinfo{number}{1} (\bibinfo{year}{1995}), \bibinfo{pages}{13--51}.
\newblock


\bibitem[\protect\citeauthoryear{Austrin}{Austrin}{2007}]%
        {austrin2007balanced}
\bibfield{author}{\bibinfo{person}{Per Austrin}.}
  \bibinfo{year}{2007}\natexlab{}.
\newblock \showarticletitle{Balanced {MAX 2-SAT} might not be the hardest}. In
  \bibinfo{booktitle}{\emph{Proc.\ 39th Symp.\ Theor.\ Comput.\ (STOC)}}.
  \bibinfo{pages}{189--197}.
\newblock


\bibitem[\protect\citeauthoryear{Austrin}{Austrin}{2010}]%
        {Austrin10}
\bibfield{author}{\bibinfo{person}{Per Austrin}.}
  \bibinfo{year}{2010}\natexlab{}.
\newblock \showarticletitle{Towards Sharp Inapproximability for any 2-{CSP}}.
\newblock \bibinfo{journal}{\emph{{SIAM} J. Comput.}} \bibinfo{volume}{39},
  \bibinfo{number}{6} (\bibinfo{year}{2010}), \bibinfo{pages}{2430--2463}.
\newblock


\bibitem[\protect\citeauthoryear{Bull and Watson}{Bull and Watson}{2007}]%
        {BullW07}
\bibfield{author}{\bibinfo{person}{Jesse Bull} {and} \bibinfo{person}{Joel
  Watson}.} \bibinfo{year}{2007}\natexlab{}.
\newblock \showarticletitle{Hard evidence and mechanism design}.
\newblock \bibinfo{journal}{\emph{Games Econom.\ Behav.}}  \bibinfo{volume}{58}
  (\bibinfo{year}{2007}), \bibinfo{pages}{75--–93}.
\newblock


\bibitem[\protect\citeauthoryear{Carroll and Egorov}{Carroll and
  Egorov}{2019}]%
        {carroll2019strategic}
\bibfield{author}{\bibinfo{person}{Gabriel Carroll} {and}
  \bibinfo{person}{Georgy Egorov}.} \bibinfo{year}{2019}\natexlab{}.
\newblock \showarticletitle{Strategic communication with minimal verification}.
\newblock \bibinfo{journal}{\emph{Econometrica}} \bibinfo{volume}{87},
  \bibinfo{number}{6} (\bibinfo{year}{2019}), \bibinfo{pages}{1867--1892}.
\newblock


\bibitem[\protect\citeauthoryear{Dughmi}{Dughmi}{2019}]%
        {dughmi2019hardness}
\bibfield{author}{\bibinfo{person}{Shaddin Dughmi}.}
  \bibinfo{year}{2019}\natexlab{}.
\newblock \showarticletitle{On the hardness of designing public signals}.
\newblock \bibinfo{journal}{\emph{Games Econom.\ Behav.}}
  \bibinfo{volume}{118} (\bibinfo{year}{2019}), \bibinfo{pages}{609--625}.
\newblock


\bibitem[\protect\citeauthoryear{Dughmi, Kempe, and Qiang}{Dughmi
  et~al\mbox{.}}{2016}]%
        {Dughmi2016limited}
\bibfield{author}{\bibinfo{person}{Shaddin Dughmi}, \bibinfo{person}{David
  Kempe}, {and} \bibinfo{person}{Ruixin Qiang}.}
  \bibinfo{year}{2016}\natexlab{}.
\newblock \showarticletitle{Persuasion with Limited Communication}. In
  \bibinfo{booktitle}{\emph{Proc.\ 17th Conf.\ Economics and Computation
  (EC)}}. \bibinfo{pages}{663--680}.
\newblock


\bibitem[\protect\citeauthoryear{Dughmi, Niazadeh, Psomas, and Weinberg}{Dughmi
  et~al\mbox{.}}{2019}]%
        {dughmi2019persuasion}
\bibfield{author}{\bibinfo{person}{Shaddin Dughmi}, \bibinfo{person}{Rad
  Niazadeh}, \bibinfo{person}{Alexandros Psomas}, {and}
  \bibinfo{person}{Matthew Weinberg}.} \bibinfo{year}{2019}\natexlab{}.
\newblock \showarticletitle{Persuasion and Incentives Through the Lens of
  Duality}. In \bibinfo{booktitle}{\emph{Proc.\ 15th Int.\ Conf.\ Web \&
  Internet Econom.\ (WINE)}}. \bibinfo{pages}{142--155}.
\newblock


\bibitem[\protect\citeauthoryear{Dughmi and Xu}{Dughmi and Xu}{2016}]%
        {dughmi2016algorithmic}
\bibfield{author}{\bibinfo{person}{Shaddin Dughmi} {and}
  \bibinfo{person}{Haifeng Xu}.} \bibinfo{year}{2016}\natexlab{}.
\newblock \showarticletitle{Algorithmic {B}ayesian persuasion}. In
  \bibinfo{booktitle}{\emph{Proc.\ 48th Symp.\ Theor.\ Comput.\ (STOC)}}.
  \bibinfo{pages}{412--425}.
\newblock


\bibitem[\protect\citeauthoryear{Dughmi and Xu}{Dughmi and Xu}{2017}]%
        {dughmi2017algorithmic}
\bibfield{author}{\bibinfo{person}{Shaddin Dughmi} {and}
  \bibinfo{person}{Haifeng Xu}.} \bibinfo{year}{2017}\natexlab{}.
\newblock \showarticletitle{Algorithmic persuasion with no externalities}. In
  \bibinfo{booktitle}{\emph{Proc.\ 18th Conf.\ Economics and Computation
  (EC)}}. \bibinfo{pages}{351--368}.
\newblock


\bibitem[\protect\citeauthoryear{Emek, Feldman, Gamzu, PaesLeme, and
  Tennenholtz}{Emek et~al\mbox{.}}{2014}]%
        {emek2014signaling}
\bibfield{author}{\bibinfo{person}{Yuval Emek}, \bibinfo{person}{Michal
  Feldman}, \bibinfo{person}{Iftah Gamzu}, \bibinfo{person}{Renato PaesLeme},
  {and} \bibinfo{person}{Moshe Tennenholtz}.} \bibinfo{year}{2014}\natexlab{}.
\newblock \showarticletitle{Signaling schemes for revenue maximization}.
\newblock \bibinfo{journal}{\emph{ACM Trans.\ Econom.\ Comput.}}
  \bibinfo{volume}{2}, \bibinfo{number}{2} (\bibinfo{year}{2014}),
  \bibinfo{pages}{1--19}.
\newblock


\bibitem[\protect\citeauthoryear{Feige and Goemans}{Feige and Goemans}{1995}]%
        {feige1995approximating}
\bibfield{author}{\bibinfo{person}{Uriel Feige} {and} \bibinfo{person}{Michel
  Goemans}.} \bibinfo{year}{1995}\natexlab{}.
\newblock \showarticletitle{Approximating the value of two power proof systems,
  with applications to {MAX 2-SAT} and {MAX DI-CUT}}. In
  \bibinfo{booktitle}{\emph{Proc.\ 3rd Israel Symp.\ Theor.\ Comput.\ Syst.}}
  \bibinfo{pages}{182--189}.
\newblock


\bibitem[\protect\citeauthoryear{Glazer and Rubinstein}{Glazer and
  Rubinstein}{2004}]%
        {GlazerR04}
\bibfield{author}{\bibinfo{person}{Jacob Glazer} {and} \bibinfo{person}{Ariel
  Rubinstein}.} \bibinfo{year}{2004}\natexlab{}.
\newblock \showarticletitle{On Optimal Rules of Persuasion}.
\newblock \bibinfo{journal}{\emph{Econometrica}} \bibinfo{volume}{72},
  \bibinfo{number}{6} (\bibinfo{year}{2004}), \bibinfo{pages}{1715--1736}.
\newblock


\bibitem[\protect\citeauthoryear{Glazer and Rubinstein}{Glazer and
  Rubinstein}{2006}]%
        {GlazerR06}
\bibfield{author}{\bibinfo{person}{Jacob Glazer} {and} \bibinfo{person}{Ariel
  Rubinstein}.} \bibinfo{year}{2006}\natexlab{}.
\newblock \showarticletitle{A study in the pragmatics of persuasion: a game
  theoretical approach}.
\newblock \bibinfo{journal}{\emph{Theor.\ Econom.}} \bibinfo{volume}{1},
  \bibinfo{number}{4} (\bibinfo{year}{2006}), \bibinfo{pages}{395--410}.
\newblock


\bibitem[\protect\citeauthoryear{Goemans and Williamson}{Goemans and
  Williamson}{1994}]%
        {GoemansW94}
\bibfield{author}{\bibinfo{person}{Michel Goemans} {and} \bibinfo{person}{David
  Williamson}.} \bibinfo{year}{1994}\natexlab{}.
\newblock \showarticletitle{.879-approximation algorithms for {MAX} {CUT} and
  {MAX} {2SAT}}. In \bibinfo{booktitle}{\emph{Proc.\ 26th Symp.\ Theor.\
  Comput.\ (STOC)}}. \bibinfo{pages}{422--431}.
\newblock


\bibitem[\protect\citeauthoryear{Gradwohl, Hahn, Hoefer, and
  Smorodinsky}{Gradwohl et~al\mbox{.}}{2020}]%
        {Gradwohl2020limited}
\bibfield{author}{\bibinfo{person}{Ronen Gradwohl}, \bibinfo{person}{Niklas
  Hahn}, \bibinfo{person}{Martin Hoefer}, {and} \bibinfo{person}{Rann
  Smorodinsky}.} \bibinfo{year}{2020}\natexlab{}.
\newblock \bibinfo{title}{Algorithms for Persuasion with Limited
  Communication}.
\newblock
\newblock
\newblock
\shownote{CoRR abs/2007.12489.}


\bibitem[\protect\citeauthoryear{Grossman}{Grossman}{1981}]%
        {grossman1981informational}
\bibfield{author}{\bibinfo{person}{Sanford Grossman}.}
  \bibinfo{year}{1981}\natexlab{}.
\newblock \showarticletitle{The informational role of warranties and private
  disclosure about product quality}.
\newblock \bibinfo{journal}{\emph{J. Law Econom.}} \bibinfo{volume}{24},
  \bibinfo{number}{3} (\bibinfo{year}{1981}), \bibinfo{pages}{461--483}.
\newblock


\bibitem[\protect\citeauthoryear{Hahn, Hoefer, and Smorodinsky}{Hahn
  et~al\mbox{.}}{2020a}]%
        {Hahn2020prophet}
\bibfield{author}{\bibinfo{person}{Niklas Hahn}, \bibinfo{person}{Martin
  Hoefer}, {and} \bibinfo{person}{Rann Smorodinsky}.}
  \bibinfo{year}{2020}\natexlab{a}.
\newblock \showarticletitle{Prophet Inequalities for Bayesian Persuasion}. In
  \bibinfo{booktitle}{\emph{Proc.\ 29th Int.\ Joint Conf.\ Artificial
  Intelligence (IJCAI)}}. \bibinfo{pages}{175--181}.
\newblock


\bibitem[\protect\citeauthoryear{Hahn, Hoefer, and Smorodinsky}{Hahn
  et~al\mbox{.}}{2020b}]%
        {Hahn2020secretary}
\bibfield{author}{\bibinfo{person}{Niklas Hahn}, \bibinfo{person}{Martin
  Hoefer}, {and} \bibinfo{person}{Rann Smorodinsky}.}
  \bibinfo{year}{2020}\natexlab{b}.
\newblock \showarticletitle{The Secretary Recommendation Problem}. In
  \bibinfo{booktitle}{\emph{Proc.\ 21st Conf.\ Economics and Computation
  (EC)}}. \bibinfo{pages}{189}.
\newblock


\bibitem[\protect\citeauthoryear{H{\aa}stad}{H{\aa}stad}{1999}]%
        {Hastad99}
\bibfield{author}{\bibinfo{person}{Johan H{\aa}stad}.}
  \bibinfo{year}{1999}\natexlab{}.
\newblock \showarticletitle{Clique is hard to approximate within
  $n^{1-\varepsilon}$}.
\newblock \bibinfo{journal}{\emph{Acta Math.}} \bibinfo{volume}{182},
  \bibinfo{number}{1} (\bibinfo{year}{1999}), \bibinfo{pages}{105--142}.
\newblock


\bibitem[\protect\citeauthoryear{Kamenica and Gentzkow}{Kamenica and
  Gentzkow}{2011}]%
        {kamenica2011bayesian}
\bibfield{author}{\bibinfo{person}{Emir Kamenica} {and}
  \bibinfo{person}{Matthew Gentzkow}.} \bibinfo{year}{2011}\natexlab{}.
\newblock \showarticletitle{Bayesian persuasion}.
\newblock \bibinfo{journal}{\emph{Amer.\ Econom.\ Rev.}} \bibinfo{volume}{101},
  \bibinfo{number}{6} (\bibinfo{year}{2011}), \bibinfo{pages}{2590--2615}.
\newblock


\bibitem[\protect\citeauthoryear{Khot and Saket}{Khot and Saket}{2016}]%
        {KhotS16}
\bibfield{author}{\bibinfo{person}{Subhash Khot} {and} \bibinfo{person}{Rishi
  Saket}.} \bibinfo{year}{2016}\natexlab{}.
\newblock \showarticletitle{Hardness of Bipartite Expansion}. In
  \bibinfo{booktitle}{\emph{Proc.\ 24th European Symp.\ Algorithms (ESA)}}.
  \bibinfo{pages}{55:1--55:17}.
\newblock


\bibitem[\protect\citeauthoryear{Letchford, Korzhyk, and Conitzer}{Letchford
  et~al\mbox{.}}{2014}]%
        {letchford2014value}
\bibfield{author}{\bibinfo{person}{Joshua Letchford}, \bibinfo{person}{Dmytro
  Korzhyk}, {and} \bibinfo{person}{Vincent Conitzer}.}
  \bibinfo{year}{2014}\natexlab{}.
\newblock \showarticletitle{On the value of commitment}.
\newblock \bibinfo{journal}{\emph{Autonomous Agents and Multi-Agent Systems}}
  \bibinfo{volume}{28}, \bibinfo{number}{6} (\bibinfo{year}{2014}),
  \bibinfo{pages}{986--1016}.
\newblock


\bibitem[\protect\citeauthoryear{Lewin, Livnat, and Zwick}{Lewin
  et~al\mbox{.}}{2002}]%
        {lewin2002improved}
\bibfield{author}{\bibinfo{person}{Michael Lewin}, \bibinfo{person}{Dror
  Livnat}, {and} \bibinfo{person}{Uri Zwick}.} \bibinfo{year}{2002}\natexlab{}.
\newblock \showarticletitle{Improved rounding techniques for the {MAX 2-SAT}
  and {MAX DI-CUT} problems}. In \bibinfo{booktitle}{\emph{Int.\ Conf.\ Integer
  Prog.\ Combinat.\ Opt.\ (IPCO)}}. \bibinfo{pages}{67--82}.
\newblock


\bibitem[\protect\citeauthoryear{Milgrom}{Milgrom}{1981}]%
        {milgrom1981good}
\bibfield{author}{\bibinfo{person}{Paul Milgrom}.}
  \bibinfo{year}{1981}\natexlab{}.
\newblock \showarticletitle{Good news and bad news: Representation theorems and
  applications}.
\newblock \bibinfo{journal}{\emph{Bell J. Econom.}} (\bibinfo{year}{1981}),
  \bibinfo{pages}{380--391}.
\newblock


\bibitem[\protect\citeauthoryear{Raghavendra}{Raghavendra}{2008}]%
        {Raghavendra08}
\bibfield{author}{\bibinfo{person}{Prasad Raghavendra}.}
  \bibinfo{year}{2008}\natexlab{}.
\newblock \showarticletitle{Optimal algorithms and inapproximability results
  for every {CSP}?}. In \bibinfo{booktitle}{\emph{Proc.\ 40th Symp.\ Theor.\
  Comput.\ (STOC)}}. \bibinfo{pages}{245--254}.
\newblock


\bibitem[\protect\citeauthoryear{Raghavendra and Steurer}{Raghavendra and
  Steurer}{2009}]%
        {RaghavendraS09a}
\bibfield{author}{\bibinfo{person}{Prasad Raghavendra} {and}
  \bibinfo{person}{David Steurer}.} \bibinfo{year}{2009}\natexlab{}.
\newblock \showarticletitle{How to Round Any {CSP}}. In
  \bibinfo{booktitle}{\emph{Proc.\ 50th Symp.\ Foundations of Computer Science
  (FOCS)}}. \bibinfo{pages}{586--594}.
\newblock


\bibitem[\protect\citeauthoryear{Raz}{Raz}{1998}]%
        {Raz98}
\bibfield{author}{\bibinfo{person}{Ran Raz}.} \bibinfo{year}{1998}\natexlab{}.
\newblock \showarticletitle{A Parallel Repetition Theorem}.
\newblock \bibinfo{journal}{\emph{{SIAM} J. Comput.}} \bibinfo{volume}{27},
  \bibinfo{number}{3} (\bibinfo{year}{1998}), \bibinfo{pages}{763--803}.
\newblock


\bibitem[\protect\citeauthoryear{Sher}{Sher}{2011}]%
        {Sher11}
\bibfield{author}{\bibinfo{person}{Itai Sher}.}
  \bibinfo{year}{2011}\natexlab{}.
\newblock \showarticletitle{Credibility and determinism in a game of
  persuasion}.
\newblock \bibinfo{journal}{\emph{Games Econom.\ Behav.}} \bibinfo{volume}{71},
  \bibinfo{number}{2} (\bibinfo{year}{2011}), \bibinfo{pages}{409--419}.
\newblock


\bibitem[\protect\citeauthoryear{Sher}{Sher}{2014}]%
        {Sher14}
\bibfield{author}{\bibinfo{person}{Itai Sher}.}
  \bibinfo{year}{2014}\natexlab{}.
\newblock \showarticletitle{Persuasion and dynamic communication}.
\newblock \bibinfo{journal}{\emph{Theor.\ Econom.}} \bibinfo{volume}{9},
  \bibinfo{number}{1} (\bibinfo{year}{2014}), \bibinfo{pages}{99--136}.
\newblock


\bibitem[\protect\citeauthoryear{Sj{\"o}gren}{Sj{\"o}gren}{2009}]%
        {sjogren2009rigorous}
\bibfield{author}{\bibinfo{person}{Henrik Sj{\"o}gren}.}
  \bibinfo{year}{2009}\natexlab{}.
\newblock \bibinfo{booktitle}{\emph{Rigorous analysis of approximation
  algorithms for MAX 2-CSP}}.
\newblock \bibinfo{publisher}{Skolan f{\"o}r datavetenskap och kommunikation,
  Kungliga Tekniska h{\"o}gskolan}.
\newblock


\bibitem[\protect\citeauthoryear{Sobel}{Sobel}{2013}]%
        {sobel2013giving}
\bibfield{author}{\bibinfo{person}{Joel Sobel}.}
  \bibinfo{year}{2013}\natexlab{}.
\newblock \showarticletitle{Giving and receiving advice}.
\newblock \bibinfo{journal}{\emph{Adv.\ Econom.\ Econometrics}}
  \bibinfo{volume}{1} (\bibinfo{year}{2013}), \bibinfo{pages}{305--341}.
\newblock


\bibitem[\protect\citeauthoryear{V{\'a}rdy}{V{\'a}rdy}{2004}]%
        {vardy2004value}
\bibfield{author}{\bibinfo{person}{Felix V{\'a}rdy}.}
  \bibinfo{year}{2004}\natexlab{}.
\newblock \showarticletitle{The value of commitment in Stackelberg games with
  observation costs}.
\newblock \bibinfo{journal}{\emph{Games and Economic Behavior}}
  \bibinfo{volume}{49}, \bibinfo{number}{2} (\bibinfo{year}{2004}),
  \bibinfo{pages}{374--400}.
\newblock


\end{thebibliography}

\begin{appendix}

\section{Missing Proofs for Constrained Equilibrium}
\subsection{Proof of Theorem~\ref{thm:equilibrium}}
\label{app: BNE Computation}
\renewcommand{\ot}{\leftarrow}
\newcommand{\bfFor}{\textbf{for}}
\newcommand{\bfDo}{\textbf{do}}
\newcommand{\bfIf}{\textbf{if}}
\newcommand{\bfThen}{\textbf{then}}
\begin{algorithm}[t]
	\caption{Bayes-Nash Equilibrium for Games with Two Actions}
	\label{algo:bayesNash_eq}
	\DontPrintSemicolon
	\KwIn{Graph $H$, distribution $q$, utility functions $u_r$ and $u_s$.}
	\KwOut{Equilibrium schemes $(\varphi,\psi)$.}
	
	Compute normalized representation of utility functions\;
	Construct initial flow network $N^0$ with edge capacities $c^0$ and set	$j \ot 0$\;
	Let $f^0$ be a maximum $s$-$t$-flow in $N^0$\;
	\While{there is an unsaturated edge $(s,\theta')$ in $N^j$}{
		Let $\Theta^j_c$ and $\Sigma^j_c$ be states and signals reachable by augementing paths from $\theta'$, resp.\;
		Set $\psi(\sigma) \ot A$ for all $\sigma \in \Sigma^j_c$\;
		
		\For{$\theta \in \Theta_{R,R} \cap \Theta_c^j$}{
			\bfIf\ \emph{$\theta$ has a neighboring signal $\sigma \in \Sigma^j \setminus \Sigma^j_c$} \bfThen\ remove $\theta$ from $\Theta_c^j$\;
		}
		\For{$\theta \in \Theta_{A,A} \cap (\Theta^j \setminus \Theta_c^j)$}{
			\bfIf\ \emph{$\theta$ has a neighboring signal $\sigma \in \Sigma^j_c$} \bfThen\ set $\varphi(\theta,\sigma) \ot c^j(s,\theta)/u_r(A,\theta)$\;
		}
		
		\bfFor\ \emph{each $\theta \in \Theta_R \cap \Theta^j_c$} \bfDo\ set $\varphi(\theta,\sigma) \ot f^j(\sigma,\theta) / u_r(R,\theta)$ for each $\sigma \in \Sigma^j_c$ \;
		\For{each $\theta \in \Theta_A \cap \Theta^j_c$}{
			Set $\varphi(\theta,\sigma) \ot f^j(\theta,\sigma) / u_r(A,\theta)$ for each $\sigma \in \Sigma^j_c$ \;
			$\varepsilon(\theta) \ot c^j(s,\theta) - \sum_{\sigma \in \Sigma^j} f^t(\theta,\sigma)$ \;
			\If{$\varepsilon(\theta) > 0$}{
				Pick arbitrary $\sigma \in \Sigma^j_c$ adjacent to $\theta$\;
				Increase $\varphi(\theta,\sigma) \ot \varphi(\theta,\sigma) + \varepsilon(\theta) / u_r(A,\theta)$
			}
		}
		Construct network $N^{j+1}$: Remove $\Theta^j_c$ and $\Sigma^j_c$ from $N^j$ along with all their incident edges.\;
		Let $f^{j+1}$ be a maximum $s$-$t$-flow in $N^{j+1}$\;
		$j \ot j+1$
	}
	
	Set $\psi(\sigma) \ot R$ for all $\sigma \in \Sigma^j$\;
	\bfFor\ \emph{each $\theta \in \Theta_A \cap \Theta^j$} \bfDo\ set $\varphi(\theta,\sigma) \ot f^j(\theta,\sigma) / u_r(A,\theta)$ for each $\sigma \in \Sigma^j$ \;
	\For{each $\theta \in \Theta_R \cap \Theta^j$}{
		Set $\varphi(\theta,\sigma) \ot f^j(\sigma,\theta) / u_r(R,\theta)$ for each $\sigma \in \Sigma^j$ \;
		$\varepsilon(\theta) \ot c^j(\theta,t) - \sum_{\sigma \in \Sigma^j} f^t(\sigma,\theta)$ \;
		\If{$\varepsilon(\theta) > 0$}{
			Pick arbitrary $\sigma \in \Sigma^j$ adjacent to $\theta$\;
			Increase $\varphi(\theta,\sigma) \ot \varphi(\theta,\sigma) + \varepsilon(\theta) / u_r(R,\theta)$
		}
	}		
	\bfFor\ \emph{each $\theta \in \Theta_N$} \bfDo\ set $\varphi(\theta,\sigma) \ot q_{\theta}$ for an incident $\sigma \in \Sigma$ as preferred by the sender\;
\end{algorithm}
%
	Our algorithm (see Algorithm~\ref{algo:bayesNash_eq} for pseudo-code) starts by adopting a normalized representation for the utility functions of the agents. For each state $\theta$, we set $\delta_r(\theta) = \min\{u_r(A,\theta), u_r(R,\theta)\}$ and adjust the utilities to $u_r(x,\theta) - \delta_r(\theta)$ for both $x \in \{A,R\}$ in each state $\theta$. In this way, we subtract $\delta_r(\theta)$ from the utility of the receiver in each state $\theta$, no matter which action is taken. Clearly, such a shift has no influence on the preference of the receiver for choosing an action. In a similar fashion we define $\delta_s(\theta)$ for the sender and shift the utilities $u_s$. In the normalized representation, for each state, an agent has positive utility for at most one action.
	
	Based on the normalized utilities, let us define the following sets of states:
	\begin{itemize}	
		\item $\Theta_{A,A} = \{ \theta \in \Theta \mid u_r(A,\theta) > 0 \text{ and } u_s(A,\theta) \ge 0 \}$
		\item $\Theta_{A,R} = \{ \theta \in \Theta \mid u_r(A,\theta) > 0 \text{ and } u_s(R,\theta) > 0 \}$
		\item $\Theta_{R,A} = \{ \theta \in \Theta \mid u_r(R,\theta) > 0 \text{ and } u_s(A,\theta) \ge 0 \}$
		\item $\Theta_{R,R} = \{ \theta \in \Theta \mid u_r(R,\theta) > 0 \text{ and } u_s(R,\theta) > 0 \}$
		\item $\Theta_A = \Theta_{A,A} \cup \Theta_{A,R}$, the (receiver) accept states
		\item $\Theta_R = \Theta_{R,A} \cup \Theta_{R,R}$, the (receiver) reject states.
		\item $\Theta_N = \{ \theta \in \Theta \mid u_r(A,\theta) = u_r(R,\theta) \}$
	\end{itemize}

	The remaining algorithm is an iterative procedure. In each round $j$, it uses a flow network $N^j = (s,t,\Theta^j,\Sigma^j,E^j, c^j)$, which is governed by the part of the graph $H$, for which signaling and action schemes have not been decided yet. In $N^j$, there is a source $s$ and a sink $t$. We start with states $\Theta^0 = \Theta \setminus \Theta_N$ and all their adjacent signals $\Sigma^0 = \{\sigma \in \Sigma \mid \text{there is } \{\theta,\sigma\} \in E \text{ with } \theta \in \Theta^0\}$. In the initial edge set $E^0$, there are directed edges $(s,\theta)$ for each $\theta \in \Theta_A$, directed edges $(\theta, \sigma)$ for each $\{\theta,\sigma\} \in (\Theta_A \times \Sigma^0) \cap E$, directed edges $(\sigma,\theta)$ for each $\{\theta,\sigma\} \in (\Sigma^0 \times \Theta_R) \cap E$ and an edge $(\theta,t)$ for each $\theta \in \Theta^R$.
	
	The flow network is a five-layer network, in which edges go from $s$ to each of the accept states, then from each accept state to every signal adjacent in $H$, and then from each signal to every adjacent reject state, and then from each reject state to the sink $t$. The initial edge capacities $c^0$ are given as follows. Each edge between states and signals has $c^j(\theta, \sigma) = \infty$ and $c^j(\sigma,\theta) = \infty$. For edges from the source $c^0(s,\theta) = q_{\theta} \cdot u_r(A,\theta)$, for the sink $c^0(\theta,t) = q_{\theta} \cdot u_r(R,\theta)$. Here we use the normalized utility values to define edge capacities.
	
	The interpretation of a flow is expected receiver utility. Consider a max-flow $f^0$ in $N^0$ and a single signal $\sigma$. Suppose $\sigma$ is sent from every state $\theta \in \Theta_A$ with a probability $f^0(\theta,\sigma)/u_r(A,\theta)$ and every $\theta \in \Theta_R$ with probability $f^0(\sigma,\theta)/u_r(R,\theta)$. Then the (unconditional) expected utility upon receiving $\sigma$ for the receiver is
	\begin{align*}
		\Ex{u_r(A,\theta) \mid \sigma } \cdot \Pr{\sigma} &= \sum_{\theta \in \Theta_A} \frac{f^0(\theta,\sigma)}{u_r(A,\theta)} \cdot u_r(A,\theta) + \sum_{\theta \in \Theta_R} \frac{f(\sigma,\theta)}{u_r(R,\theta)} \cdot u_r(A,\theta) = \sum_{\theta \in \Theta_A} f^0(\theta,\sigma) \\
		\Ex{u_r(R,\theta) \mid \sigma } \cdot \Pr{\sigma} &= \sum_{\theta \in \Theta_A} \frac{f^0(\theta,\sigma)}{u_r(A,\theta)} \cdot u_r(R,\theta) + \sum_{\theta \in \Theta_R} \frac{f(\sigma,\theta)}{u_r(R,\theta)} \cdot u_r(R,\theta) = \sum_{\theta \in \Theta_R} f^0(\sigma,\theta)
	\end{align*}
	Hence, by flow conservation, such an assignment yields $\Ex{u_r(A,\theta) \mid \sigma } = \Ex{u_r(R,\theta) \mid \sigma }$, i.e., it leaves the receiver indifferent between both choices upon receiving signal $\sigma$. 
	
	In each iteration $j$, our algorithm checks if the max-flow leaves an edge $(s,\theta')$ to a state $\theta' \in \Theta^j \cap \Theta_A$ unsaturated. This implies that there is a part of the graph $H$, in which the expected utility from the accept side exceeds the one from the reject side. The part of the graph is identified by states $\Theta_c^j$ and $\Sigma_j^c$ considering augmenting paths from $\theta'$. This ensures that for each signal $\sigma \in \Sigma_j^c$ all states that route flow from or to $\sigma$ are included in $\Theta^j_c$. Then by assigning $\varphi(\theta,\sigma) = f^j(\theta,\sigma)/u_r(A,\theta)$ for $\theta \in \Theta^j_c \cap \Theta_A$ and $\varphi(\theta,\sigma) = f^j(\sigma,\theta)/u_r(R,\theta)$ for $\theta \in \Theta^j_c \cap \Theta_R$, we obtain a ``break-even'' assignment with the same conditional expectations for the receiver utility of both actions as above. 
	
	We then have to incorporate additional probability mass from $\theta'$ and other states in $\Theta_A$, which tilts the preference of the receiver to action $A$. Consequently, we assign each signal $\sigma \in \Sigma^j_c$ to $\psi(\sigma) = A$. If there are states $\theta \in \Theta^j_c \cap \Theta_A$ that are not saturated, we add the remaining probability mass to an arbitrary adjacent signal from $\Sigma^c_j$. This only increases the preference of the receiver for $A$. For a state $\theta \in \Theta_{R,R}$, we include $\theta$ in our assignment only if all remaining adjacent signals are in $\Theta_j^c$. This implies that all adjacent signals (from $\Theta^j_c$ in current and potentially previous iterations) are accept signals. Otherwise, removing the contribution of $\theta$ only increases the receiver preference for $A$ for the adjacent signals. Finally, for states $\theta \in \Theta_{A,A}$, both sender and receiver want that $A$ is chosen. As such, whenever a signal $\sigma \in \Sigma^j_c$ has an adjacent state $\theta \in \Theta_{A,A}$, then we always signal $\sigma$ in state $\theta$. This can only increase the preference of the receiver for $A$ in $\sigma$.
	
	Now suppose the while-loop of the algorithm breaks. Since the remaining network (if any) has no unsaturated edge $(s,\theta)$, a ``break-even'' assignment of $\varphi(\theta,\sigma) = f^j(\theta,\sigma)/u_r(A,\theta)$ for $\theta \in \Theta^j_c \cap \Theta_A$ and $\varphi(\theta,\sigma) = f^j(\sigma,\theta)/u_r(R,\theta)$ for $\theta \in \Theta^j_c \cap \Theta_R$ completely assigns all probability mass of the remaining accept states in $\Theta^j \cap \Theta_A$. As such, we can only have an excess utility on the reject side. We simply add any excess probability from the remaining states in $\Theta^j \cap \Theta_R$ arbitrarily to the signals in $\Sigma^j$. We then turn every signal $\sigma \in \Sigma^j$ into a reject signal $\psi(\sigma) = R$. This clearly is aligned with the preference of the receiver. 
	
	In the final step, we assign each $\theta \in \Theta_N$ according to the preference of the sender for the actions of the adjacent signals. In $\theta$, the receiver is completely indifferent between both actions. As such, an assignment of $\theta$ has no influence on the preference of the receiver for an action upon receiving a signal. Overall, this shows that the assignment of actions $\psi$ for all signals is aligned with the receiver preferences, i.e., $\psi$ is a best response against $\varphi$.
	
	Now consider the incentives for the sender. For every state $\theta \in \Theta_{A,A} \cup \Theta_{R,A}$ for which $\varphi(\theta,\sigma)$ is assigned during the while-loop, the scheme only sends accept signals. This is clearly in the interest of the sender. For a state $\theta \in \Theta_{R,R}$, $\varphi(\theta,\sigma)$ gets assigned during the while-loop only when all adjacent signals are accept signals in $\psi$. As such, the sender must send an accept signal in $\theta$, and $\varphi$ is a best response against $\psi$. Consider a state $\theta \in \Theta_{A,R}$ assigned during the while loop. In the iteration $j$, in which $\varphi(\theta,\sigma)$ gets assigned, we have $\theta \in \Theta_j^c$. Hence, every adjacent signal $\sigma$ either (a) was decided to be an accept signal in an earlier iteration, or (b) becomes an accept signal in this iteration, because $\sigma \in \Sigma^j_c$ -- any augementing path to $\theta$ can be extended along $(\theta,\sigma)$ with infinite capacity. Hence, all adjacent signals $\sigma$ have $\psi(\sigma) = A$, and $\varphi(\theta,\sigma)$ sending only accept signals is a best response against $\psi$. 
	
	Now consider the states assigned to reject signals after the while-loop breaks. For any state $\theta \in \Theta_{R,R} \cup \Theta_{A,R}$, sending reject signals is in the interest of the sender, and hence $\varphi$ is a best response. If a state $\theta \in \Theta_{A,A}$ is assigned after the loop, it has no adjacent accept signal, since otherwise it would have been assigned to that signal before. As such, all adjacent signals are reject signals, and $\varphi$ is a best response. Now consider a state $\theta \in \Theta_{R,A}$. Suppose $\theta$ has an adjacent accept signal $\sigma$ and consider the iteration $j$ where $\sigma \in \Sigma^j_c$. Then we can extend the augmenting path from $\sigma$ via $(\sigma, \theta)$, so $\theta \in \Theta^j_c$, and $\varphi$ for $\theta$ must have been assigned during iteration $j$. This is a contradiction -- as such, if $\theta \in \Theta_{R,A}$ is assigned after the while-loop, all adjacent signals must be reject signals under $\psi$. $\varphi$ is again a best response. Finally, every state $\theta \in \Theta_N$ is assigned according to the preference of the sender to adjacent signals. Overall, this shows that the assignment of $\varphi$ for all states is aligned with the sender preferences, i.e., $\varphi$ is a best response against $\psi$.

\section{Missing Proofs for Constrained Delegation}

\subsection{Proof of Lemma~\ref{lem: GR 06 deterministic psi}}\label{app: proof of GR 06}

In constrained delegation, for every state $\theta$ the sender always picks the signal that maximizes the probability to accept. Consider any optimal scheme $\psi^*$. For every signal $\sigma \in \Sigma$, $\psi^*(\sigma)$ is the probability that the receiver accepts in $\psi^*$. We set $p^{max} = \max_{\sigma} \psi^*(\sigma)$ and $p^{min} = \min_{\sigma} \psi^*(\sigma)$.

Consider the set $\Sigma_1 = \{ \sigma \in \Sigma \mid \psi^*(\sigma) = p^{max}\}$ of signals with largest acceptance probability and the set $\Sigma_2 = \{ \sigma \in \Sigma \mid \psi^*(\sigma) = \max_{\sigma \in \Sigma \setminus \Sigma_1} \psi^*(\sigma)\}$ with second-largest acceptance probability $p^{2nd} = \psi^*(\sigma)$ for $\sigma \in \Sigma_2$. If $\Sigma_2$ is empty, we let $p^{2nd} = 0$.

If $\Ex{u_r(A, \theta) \mid \sigma \in \Sigma_1} \ge \Ex{u_r(R, \theta) \mid \sigma \in \Sigma_1}$, then it is profitable for the receiver to raise all probabilities $\psi^*(\sigma)$ of signals in $\sigma \in \Sigma_1$ to 1. After this step, the signals in $\Sigma_1$ remain to ones with largest acceptance probability. Hence, this does not change the preferences of the sender and the resulting assignment of states of nature to signals. Otherwise, it is profitable to lower all probabilities $\psi^*_{\sigma}$ of signals in $\sigma \in \Sigma_1$ down to $p^{2nd}$. As long as the probability stays strictly above $p^{2nd}$, the signals in $\Sigma_1$ remain to ones with largest acceptance probability. As such, this does not change the preferences of the sender and the resulting assignment of states of nature to signals. When the probability becomes equal to $p^{2nd}$, the set $\Sigma_1$ is joined with $\Sigma_2$. At this point the sender might change the assignment due to different tie-breaking among signals in $\Sigma_1 \cup \Sigma_2$. However, we assume that the tie-breaking is executed in favor of the receiver. As such, the resulting scheme becomes even more profitable for the receiver.

Applying this procedure iteratively, we see that the probabilities for signals in $\Sigma_1$ are either raised to 1 or lowered to $p^{2nd}$. In the former case, we proceed to $\Sigma_2$ and apply the same argument. In the latter case, we proceed with $\Sigma_1 \cup \Sigma_2$ and repeat the argument. This eventually leads to an optimal deterministic assignment with all probabilities in $\{0,1\}$. \qed

\subsection{Proof of Theorem~\ref{thm:2-inapprox-dsv}}\label{app: lower bound for delegation}

In the \textsc{Colored Bipartite Vertex Expansion (CBVE)} problem, we are given a bipartite graph $(U, V, E)$ where the left vertex set $U$ is partitioned into $U_1 \cup \cdots \cup U_k$; we refer to each $U_i$ as a \emph{color class}. A subset $S \subseteq U$ is said to be \emph{colorful} iff $|S \cap U_i| \leq 1$ for all $i \in [k]$.  The goal of \textsc{CBVE} is to find a colorful subset $S \subseteq U$ of a given size such that $N(S)$ is minimized.

In this section, we will prove \classNP-hardness of approximating constrained delegation (Theorem~\ref{thm:2-inapprox-dsv}). The reduction is divided into two main parts. First, we show the \classNP-hardness of approximating \textsc{CBVE} (Theorem~\ref{thm:cbvx}), akin to Khot and Saket's hardness of \BVE\ (Theorem~\ref{thm:bvx}). This is done in the following two subsections. Then, we reduce from \CBVE\ to constrained delegation in Section~\ref{sec:cbvx-to-simple-del}; this reduction is similar to that from \BVE\ sketched in Section~\ref{subsec: hardness of delegation}.

Since we often deal with multiple graphs in this section, we may write $N_G$ instead of $N$ to stress that we are referring to the neighborhood set in graph $G$ to avoid any ambiguity.

\subsubsection{From Label Cover to Colored Bipartite Vertex Expansion}

We will prove the following \classNP-hardness of \CBVE. We remark that this is not yet the final hardness we use to reduce to constrained delegation yet; in particular, unlike Theorem~\ref{thm:bvx}, the NO case can still have $\alpha/t \ll 1$. We will ``boost'' the NO case so that the coefficient is arbitrarily close to 1 in the next subsection.

\begin{theorem} \label{thm:cbvx-intermediate}
For any constants $\tau \in (0, 1)$ and $\alpha > 1$, there exists $t = t(\tau, \alpha)$ such that, given a bipartite graph $(U, V, E)$ together with a partition $U = U_1 \cup \cdots \cup U_k$, it is \classNP-hard to distinguish between the following two cases:
\begin{itemize}
\item (YES) There exists a colorful $S^* \subseteq U$ of size $k$ such that $|N(S^*)| = \frac{1}{t} \cdot |V|$.
\item (NO) For any colorful $S \subseteq U$ of size at least $\tau k$, we have $|N(S)| > \frac{\alpha}{t} \cdot |V|$.
\end{itemize}
\end{theorem}

To prove Theorem~\ref{thm:cbvx-intermediate}, we reduce from the Label Cover problem, a canonical problem used as a starting point in numerous hardness of approximation results. Below we summarize the definition and hardness of Label Cover needed for our purpose.

\begin{definition}[Label Cover]
A \emph{Label Cover} instance $\calL = (A, B, E, \{\pi\}_{e \in E}, \Lambda)$ consists of
\begin{itemize}
\item a bi-regular bipartite graph $(A, B, E)$, which we refer to as the \emph{constraint graph},
\item the \emph{label set}  $\Lambda$,
\item for each edge $e \in E$, the \emph{constraint} (or \emph{projection}) $\pi_e: \Lambda \to \Lambda$.
\end{itemize}
We say that an assignment $\phi: (A \cup B) \to \Lambda$ \emph{satisfies} an edge $(a, b) \in E$ iff $\pi_{(a, b)}(\phi(a)) = \phi(b)$. The goal is to find an assignment that satisfies as large a fraction of edges as possible.
\end{definition}

\begin{theorem}[\cite{Raz98}] \label{thm:label-cover}
For every constant $\varepsilon > 0$, there exists $t = t(\varepsilon)$ such that, given a Label Cover instance $\calL$ with $|\Lambda| = t$, it is \classNP-hard to distinguish between the following two cases:
\begin{itemize}
\item (YES) There exists an assignment that satisfies all the edges. 
\item (NO) Every assignment satisfies less than an $\varepsilon$-fraction of the edges.
\end{itemize}
\end{theorem}

We now prove Theorem~\ref{thm:cbvx-intermediate}. The proof is relatively simple and is based on a viewpoint of the whole Label Cover instance $\calL$ as a so-called \emph{labelled-extended graph}, where the left vertex set is $A \times \Lambda$, the right vertex set is $B \times \Lambda$ and the edges are defined naturally based on the constraints. It is not hard to see that, in the YES case, picking a subset according to satisfying assignment results in a subset that does not expand much into the right vertex set. The NO case can also be argued as expected. We formalize this intuition below.

\begin{proof}[Proof of Theorem~\ref{thm:cbvx-intermediate}]
Suppose $\eps = \tau^2 / \alpha$ and $t = t(\eps)$ as in Theorem~\ref{thm:label-cover}. Consider an instance of Label Cover $\calL = (A', B', E', \{\pi_e\}_{e \in E'}, \Lambda)$  where $|\Lambda| = t$. We construct an instance $G = (U = U_1 \cup \cdots \cup U_k, V, E)$ of \CBVE\ as follows.
\begin{itemize}
\item $U = A' \times \Lambda$ and $V = B' \times \Lambda$.
\item Add an edge between $(a, \lambda_a) \in U$ and $(b, \lambda_b) \in V$ to $E$ iff $(a, b) \in E'$ and $\pi_{(a, b)}(\lambda_a) = \lambda_b$.
\item $k = |A'|$. Rename the vertices in $A'$ as $1, \dots, k$. The $i$-th color class is given by $U_i = \{i\} \times \Lambda$.
\end{itemize}

We will now prove correctness of the reduction. First, it is obvious that the reduction can be implemented in polynomial time. Below, we will show completeness (i.e., the YES case in Theorem~\ref{thm:label-cover} results in the YES case in Theorem~\ref{thm:cbvx-intermediate}) and soundness (i.e., the NO case in Theorem~\ref{thm:label-cover} results in the NO case in Theorem~\ref{thm:cbvx-intermediate}). Together with Theorem~\ref{thm:label-cover}, these complete the proof of Theorem~\ref{thm:cbvx-intermediate}.

\paragraph{(Completeness)}
Suppose that there exists an assignment $\phi^*: (A' \cup B') \to \Lambda$ that satisfies all the edges in $\calL$. Let $S^* = \{(a, \phi^*(a)) \mid a \in A'\}$. Since $\phi^*$ satisfies all edges in $\calL$, we have $N_G(S^*) = \{(b, \phi^*(b)) \mid b \in B'\}$. As a result, we have $|N_G(S^*)| = |B'| = \frac{|V|}{t}$ as desired.

\paragraph{(Soundness)}
We will prove this contrapositively. Specifically, assume that there exists $S \subseteq U$ of size at least $\tau k$ with $|N_G(S)| \leq \frac{\alpha}{t} \cdot |V| = \alpha \cdot |B'|$. We will show that there exists an assignment that satisfies at least an $\eps$-fraction of the edges in $\calL$.

For convenience, let $H = (A', B', E')$, and let us denote by $d_{A'}$ and $d_{B'}$ the degree of each vertex in $A'$ and the degree of each vertex in $B'$, respectively.

We define $T = \{a \in A' \mid \exists \lambda \in \Lambda, (a, \lambda) \in S\}$. Now since $S$ is colorful, we must have $|T| = |S| \geq \tau k = \tau |A'|$. Furthermore, for every $b \in N_H(T)$, let $\Lambda(b) = \{\lambda \in \Lambda \mid (b, \lambda) \in N_G(T)\}$. We define a random assignment $\phi: (A' \cup B') \to \Lambda$ as follows:
\begin{itemize}
\item For every $a \in T$, let $\phi(a)$ be the unique label such that $(a, \phi(a)) \in S$. (The uniqueness is due to colorfulness of $S$.)
\item For every $b \in T$, let $\phi(b)$ be a random label from $\Lambda(b)$.
\item For other vertices $c \in (A' \cup B') \setminus (T \cup N_G(T))$, let $\phi(c)$ be an arbitrary label from $\Lambda$.
\end{itemize}
Observe that, for $a \in T$ and $b \in N_H(a)$, the probability that $(a, b)$ is satisfied by $\phi$ is exactly $\frac{1}{|\Lambda(b)|}$. Hence, the expected fraction of constraints satisfied by $\phi$ is at least
\begin{align*}
\frac{1}{|E'|} \cdot \sum_{a \in T} \sum_{b \in N_H(a)} \frac{1}{|\Lambda(b)|} 
&\geq \frac{1}{|E'|} \cdot \frac{\left(\sum_{a \in T} \sum_{b \in N_H(a)} 1\right)^2}{\left( \sum_{a \in T} \sum_{b \in N_H(a)} |\Lambda(b)|\right)} \\
&= \frac{1}{|E'|} \cdot \frac{\left(|T| \cdot d_{A'}\right)^2}{\left( \sum_{b \in N_H(T)} \sum_{a \in N_H(b) \cap T} |\Lambda(b)|\right)} \\
&\geq \frac{1}{|E'|} \cdot \frac{\left(\tau |A'| \cdot d_{A'}\right)^2}{d_{B'} \left(\sum_{b \in N_H(T)} |\Lambda(b)|\right)} \\
&= \frac{1}{|E'|} \cdot \frac{\left(\tau |E'|\right)^2}{d_{B'} \cdot |N_G(S)|} \\
&\geq \frac{1}{|E'|} \cdot \frac{\left(\tau |E'|\right)^2}{d_{B'} \cdot \alpha \cdot |B'|} \\
&= \frac{\tau^2}{\alpha} \\
&= \varepsilon,
\end{align*}
where the first inequality follows from the Cauchy-Schwarz inequality, the second inequality from $|T| \geq \tau|A'|$, and the third inequality from $|N_G(S)| \leq \alpha \cdot |B'|$.

Hence, we can conclude that there exists an assignment that satisfies at least an $\eps$-fraction of the edges in $\calL$ as desired. 
\end{proof}

\subsubsection{Amplifying Completeness of \CBVE}

Our next step is to translate the hardness from Theorem~\ref{thm:cbvx-intermediate} into a form similar to that of Theorem~\ref{thm:bvx}. Specifically, we have to ``boost'' the NO case so that $|N(S)|$ is at least $(1 - \gamma)|V|$. We give a more precise statement below.

\begin{theorem} \label{thm:cbvx}
For any constants $\tau, \gamma \in (0, 1)$, given a bipartite graph $(U, V, E)$ together with a partition $U = U_1 \cup \cdots \cup U_k$, it is \classNP-hard to distinguish between the following two cases:
\begin{itemize}
\item (YES) There exists a colorful $S^* \subseteq U$ of size $k$ such that $|N(S^*)| \leq \gamma |V|$.
\item (NO) For any colorful $S \subseteq U$ of size at least $\tau k$, we gave $|N(S)| > (1 - \gamma)|V|$.
\end{itemize}
\end{theorem}

The proof of Theorem~\ref{thm:cbvx} follows a standard technique of using graph products to amplify gaps. In particular, it is almost the same as what is referred to as the ``OR-product'' in~\cite{KhotS16}, except that we only apply it on one vertex set. We provide the full argument below.

\begin{proof}[Proof of Theorem~\ref{thm:cbvx}]
We extend our reduction in the proof of Theorem~\ref{thm:cbvx-intermediate}. Let $\alpha = \frac{100}{\gamma} \cdot \ln(10/\gamma)$ and let $t = t(\tau, \alpha)$ be as in Theorem~\ref{thm:cbvx-intermediate}. Furthermore, we set $\ell = \lfloor \gamma t \rfloor$.

Suppose $H = (U = U_1 \cup \cdots \cup U_k, V', E')$ is the hard \CBVE\ instance from Theorem~\ref{thm:cbvx-intermediate}. We create a new \CBVE\ instance $G = (U = U_1 \cup \cdots \cup U_k, V, E)$ as follows.
\begin{itemize}
\item $U$ and its partition $U_1 \cup \cdots \cup U_k$ are the same as in the original instance.
\item $V = (V')^{\ell}$ is the set of all $\ell$-tuples of vertices in $V'$.
\item Add an edge in $E$ between $u \in U$ and $(v_1, \dots, v_\ell) \in V$ iff $(u, v_i) \in E'$ for at least one $i \in [\ell]$.
\end{itemize}

It is obvious that the reduction can be implemented in polynomial time. Before we prove the completeness and soundness of the reduction, let us observe that the following identity holds for all $S \subseteq U$:
\begin{align} \label{eq:or-product-bound}
\frac{|N_G(S)|}{|V|} \quad = \quad 1 - \left(1 - \frac{|N_H(S)|}{|V'|}\right)^{\ell}\enspace.
\end{align}
It is simple to check that the above identity holds, because a vertex $(v_1, \dots, v_\ell) \in V$ does not belong to $N_G(S)$ iff $(v_1, \dots, v_\ell) \in (V' \setminus N_H(S))^{\ell}$. With this identity in mind, we now proceed to prove the completeness and soundness of the reduction.

\paragraph{(Completeness)}
Suppose that there exists a colorful $S^* \subseteq U$ such that $|N_H(S)| = \frac{1}{t} \cdot |V'|$. From~\eqref{eq:or-product-bound}, we have
\begin{align*}
\frac{|N_G(S^*)|}{|V|}\quad = \quad 1 - \left(1 - \frac{1}{t}\right)^{\ell} \quad \overset{\substack{\text{Bernoulli's}\\\text{inequality}}}{\leq} \quad \frac{\ell}{t} \quad \leq \quad \gamma\enspace.
\end{align*}
In other words, $|N_G(S^*)| \leq \gamma |V|$ as desired.

\paragraph{(Soundness)}
Suppose that any colorful set $S \subseteq U$ satisfies $|N_H(S)| > \frac{\alpha}{t} \cdot |V'|$. From~\eqref{eq:or-product-bound}, we also have
\begin{align*}
\frac{|N_G(S)|}{|V|} \quad > \quad 1 - \left(1 - \frac{\alpha}{t}\right)^{\ell}\enspace.
\end{align*}
By our choice of $\alpha$ and soundness of Theorem~\ref{thm:cbvx-intermediate}, we must have $t \geq \frac{1}{\alpha} > 2 / \gamma$, meaning that $\ell \geq \gamma t / 2$. Plugging this into the above inequality,
\begin{align*}
\frac{|N_G(S)|}{|V|}  \quad > \quad 1 - \left(1 - \frac{\alpha}{t}\right)^{\gamma t / 2} \quad \geq  \quad 1 - (e^{-\alpha/t})^{\gamma t / 2}  \quad \geq  \quad 1 - \gamma\enspace,
\end{align*}
where the last inequality follows from our choice of $\alpha$. This completes our proof. 
\end{proof}

\subsubsection{From \CBVE\ to Constrained Delegation}
\label{sec:cbvx-to-simple-del}

Finally, we reduce from the \classNP-hardness of \CBVE\ in Theorem~\ref{thm:cbvx} to the \classNP-hardness of constrained delegation (Theorem~\ref{thm:2-inapprox-dsv}). The reduction closely follows the sketch in Section~\ref{subsec: hardness of delegation}, except that the acceptable states are now used to check the ``colorfulness'' of $S$ instead of its size.

\begin{proof}[Proof of Theorem~\ref{thm:2-inapprox-dsv}]
For any constant $\varepsilon \in (0, 1)$, set $\gamma = \tau = 0.1\varepsilon$. Let $(U = U_1 \cup \cdots \cup U_k, V, E)$ be the input to the \BVE\ problem. We construct an instance of constrained delegation as follows.
\begin{itemize}
\item For every vertex $u \in U$, create a signal $\sigma_u$.
\item The set of rejectable states is $\Theta_R = \{\theta_v \mid v \in V\}$. For each $\theta_v \in \Theta_R$, its set of allowed signals is $N(\theta_v) = \{ \sigma_u \mid u \in N(v) \}$. The probability is $q_{\theta_v} = \frac{1}{2|V|}$.
\item The set of acceptable states is $\Theta_A = \{\theta_i \mid i \in [k]\}$. For each $\theta_i \in \Theta_A$, its set of allowed signals is $N(\theta_i) = U_i$. The probability is $q_{\theta_i} = \frac{1}{2k}$.
\end{itemize}
Observe that $q_A = q_R = 0.5$.

It is obvious to see that the above reduction can be implemented in polynomial time. We will now prove the completeness and soundness properties of our reduction. Specifically, let $\OPT = 0.5 \cdot (2 - \gamma)$; we will show below that the YES case (of Theorem~\ref{thm:cbvx}) results in a decision scheme with utility at least $\OPT$, whereas the NO case implies that any decision scheme has utility less than $\frac{\OPT}{2 - \varepsilon}$. Note that this, together with Theorem~\ref{thm:cbvx}, completes the proof of Theorem~\ref{thm:2-inapprox-dsv}.

\paragraph{(Completeness)}
Suppose that there exists a colorful $S^* \subseteq U$ of size $k$ such that $|N(S^*)| \leq \gamma |V|$. Consider the (deterministic) decision scheme $\psi^*$ where $\psi^*(\sigma_u) = 1$ iff $u \in S^*$. Since $S^*$ is colorful and has size $k$, every acceptable state is accepted. On the other hand, a rejectable state $\theta_v \in \Theta_R$ is accepted iff $v \in N(S^*)$. Hence, the utility of $\psi^*$ is at least
\begin{align*}
\frac{1}{2|V|}(|V| - |N(S^*)|) + \frac{1}{2} \quad \geq \quad  0.5(2 - \gamma) \quad = \quad \OPT\enspace,
\end{align*}
where the inequality follows from $|N(S^*)| \leq \gamma |V|$.

\paragraph{(Soundness)}
Suppose that, for any colorful set $S \subseteq U$ of size at least $\tau k$, we have $|N(S)| > (1 - \gamma) |V|$. Consider an optimal decision scheme $\psi$. We will show that the utility achieved by $\psi$ is at most $\frac{OPT}{2 - \varepsilon}$. From Lemma~\ref{lem: GR 06 deterministic psi}, we may assume that $\psi$ is deterministic, i.e., $\psi(\sigma) \in \{0, 1\}$ for all $\sigma \in \Sigma$. Observe further that, if there exist distinct vertices $u, u'$ from the same color class $U_i$ such that $\psi(\sigma_u) = \psi(\sigma_{u'}) = 1$, then we may modify $\psi(\sigma_u)$ to zero without decreasing the utility\footnote{Specifically, the utility with respect to $\Theta_A$ remains the same, whereas the utility with respect to $\Theta_R$ does not decrease.}. In other words, we may assume that $\sigma_u = 1$ for at most one vertex $u$ in each color class $U_i$. Let $S = \{u \in U \mid \psi(\sigma_u) = 1\}$. The aforementioned assumption implies that $S$ is colorful. We consider two cases, based on whether $|S| \geq \tau k$.
\begin{itemize}
\item Case I: $|S| \geq \tau k$.

In this case, from our assumption, we must have $|N(S)| > (1 - \gamma)|V|$. Since every rejectable state $\sigma_v$ for $v \in N(S)$ is accepted, the utility of $\psi$ is at most
\begin{align*}
\frac{1}{2|V|}(|V| - |N(S)|) + \frac{1}{2} \quad < \quad 0.5(1 + \gamma) \quad \leq \quad \frac{\OPT}{2 - \eps}\enspace,
\end{align*}
where the last inequality follows with our choice of $\gamma$.
\item Case II: $|S| < \tau k$. 

In this case, at most $|S|$ acceptable states are accepted; this means that the utility of $\psi$ is at most
\begin{align*}
\frac{1}{2} + \frac{|S|}{2k} \quad < \quad 0.5(1 + \tau) \quad \leq \quad \frac{\OPT}{2 - \eps}\enspace,
\end{align*}
where the last inequality follows with our choice of $\tau$.
\end{itemize}
Hence, in both cases, the utility of the decision scheme is at most $\frac{OPT}{2 - \varepsilon}$ as desired. 
\end{proof}

\subsection{APX-hardness with Degree-2 States and Degree-1 Rejects}\label{app:APX}

The following result is a consequence of~\cite[Theorem 7]{Sher14}.

\begin{corollary}
  \label{cor:APX}
  Constrained delegation is \classAPX-hard for instances with degree-2 states and degree-1 rejects.
\end{corollary}

\begin{proof}
Consider an instance of the \VC\ problem given by an undirected graph $G=(V,E)$. For each vertex $v \in V$ we introduce a signal $\sigma_v$. For every edge $e \in E$ we introduce an acceptable state of nature $\theta_e$, i.e., $\Theta_A = \{ \theta_e \mid e\in E \}$. For every vertex $v \in V$ we introduce a rejectable state of nature $\theta_v$. For every $e = \{v,w\}$ the state $\theta_e$ has two allowed signals $\sigma_v, \sigma_w$. For every vertex $v$ the state $\theta_v$ has the allowed signal $\sigma_v$. The distribution over states is the uniform distribution, i.e., $q_{\theta} = 1/(|E| + |V|)$ for every $\theta \in \Theta$.

We can restrict the optimal scheme $\psi^*$ to be deterministic. For an accept signal the acceptance probability is 1, for a reject signal it is 0. To ensure the correct action in $\theta_v$, signal $\sigma_v$ must be a reject signal. To ensure the correct action in $\theta_e$, at least one incident signal $\sigma_v, \sigma_w$ must be an accept signal. Now consider any subset $\Sigma'$ of accept signals and the corresponding subset $V'$ of vertices in $G$. For this subset the expected utility for $R$ is
\[ 
\frac{1}{|E| + |V|} \cdot (|E(V')| + (|V| - |V'|))
\]
where $E(V')$ is the set edges incident to at least one vertex in $V'$. 

For an edge $e$, suppose there is no incident vertex in $V'$. Then adding one (say $v$) to $V'$ can only increase the profit ($\theta_e$ action becomes correct, $\theta_v$ action becomes wrong). Hence, w.l.o.g.\ we assume $E(V') = E$, i.e., $V'$ is a vertex cover. As such, the optimal decision scheme $\psi^*$ has a profit of at least $\alpha$ if and only if $G$ has a vertex cover of size at most $(1-\alpha)(|E| + |V|)$. This proves \classNP-hardness.

For \classAPX-hardness, we consider 3-regular graphs with $|E| = 1.5|V|$, where every vertex cover $V'$ has size $|V'| = \Theta(|V|)$. In these graphs, \VC\ is \classAPX-hard~\cite{AlimontiK00}. Therefore, with our reduction we also obtain a constant hardness of approximation for the objective \[|E| + |V| - |V'| = 2.5|V| - |V'|. \]
\end{proof}

\section{Hardness for Constrained Persuasion}\label{app: persuasion hardness}

We first prove the approximation hardness for the general case, which applies in the cases of degree-1 accepts and degree-2 states (Theorem~\ref{thm:verifyHard}). We then also prove the result for degree-1 rejects (Theorem~\ref{thm:verifyHard2}).

\subsection{Proof of Theorem~\ref{thm:verifyHard}}
We build a reduction from the \IS\ problem. In this problem, we are given an undirected graph $G = (V,E)$. An independent set is a subset $I \subseteq V$ of the vertices such that no two vertices in $I$ are connected via an edge from $E$. The goal is to find an independent set of maximum cardinality. Without loss of generality, we assume there are no isolated vertices in $G$, since these vertices are trivially in the optimum solution. 


Hastad~\cite{Hastad99} proved that it is \classNP-hard to approximate the maximum independent set problem to within a factor of $|V|^{1 - \varepsilon}$ for any constant $\varepsilon > 0$.
For a given instance $G$ of \IS\ we build a constrained persuasion problem such that the optimum utility of the sender is proportional to the cardinality of the largest independent set in $G$. 
As a consequence, constrained persuasion cannot be approximated within a factor of $|V|^{1-\varepsilon}$, for any constant $\varepsilon > 0$.

Our construction works as follows. For every vertex $v \in V$, we introduce an acceptable state of nature $\theta_v$ with probability $q_{\theta_v} = \frac{1}{|V| + 3|E|}$. For every edge $e \in E$ we introduce an rejectable state $\theta_e$ with probability $q_{\theta_e} = \frac{3}{|V| + 3|E|}$. For every vertex $v \in V$ we introduce a signal $\sigma_v$. Note that $n = |\Sigma| = |V|$ and $m = |\Theta| = |V| + |E|$. For the state-signal graph $H$, we insert an edge $(\theta_v, \sigma_v)$ for every $v \in V$. Moreover, we add $(\theta_e, \sigma_v)$ iff $v$ is incident to $e$. As such, in state $\theta_v$ we are forced to signal $\sigma_v$. In state $\theta_e$, we can choose from two signals corresponding to the incident vertices of $e$.

Hence, in any signaling scheme $\varphi$, we only need to determine $\varphi(\theta_e,\sigma_v)$ for one vertex $v$ incident to $e$. 
Observe that, for any edge $(v, w) \in E$, $\sigma_v$ and $\sigma_w$ cannot both imply an accept decision for the receiver because
\begin{multline*}
\sum_{\theta \in N(\sigma_v) \cap \Theta_A} \varphi(\theta,\sigma_v) + \sum_{\theta \in N(\sigma_w) \cap \Theta_A} \varphi(\theta,\sigma_w)
= \frac{2}{|V| + 3|E|} < \frac{3}{|V| + 3|E|} = q_{\theta_e} \\
\leq \sum_{\theta \in N(\sigma_v) \setminus \Theta_A} \varphi(\theta, \sigma_v) + \sum_{\theta \in N(\sigma_w) \setminus \Theta_A} \varphi(\theta, \sigma_w).
\end{multline*}

%
%

Hence, the set of accept signals $\Sigma_A$ always represents an independent set in $G$. Moreover, for any accept signal $\sigma_v$, it holds that 
\[\sum_{\theta \in N(\sigma_v)} \varphi(\theta, \sigma_v) \le 2 \sum_{\theta \in N(\sigma_v) \cap \Theta_A} \varphi(\theta, \sigma_v) = \frac{2}{|V| + 3|E|}\enspace. \]
Thus, the maximum utility that the sender can obtain is at most $|I^*| \cdot \frac{2}{|V| + 3|E|}$, where $I^*$ is a maximum independent set in $G$. Finally, we construct a simple optimal signaling scheme $\varphi^*$ based on $I^*$ using which the sender obtains this maximal utility. For every vertex $v \in I^*$, we pick one incident edge $e = \{v,w\}$ and set $\varphi^*(\theta_e, \sigma_v) = \frac{1}{3} q_{\theta_e}$ and $\varphi^*(\theta_e, \sigma_w) = \frac{2}{3} q_{\theta_e}$ (since $w \not\in I^*$ by construction). For all other edges $e \in E$, we set $\varphi(\theta_e, \sigma_w) = q_{\theta_e}$ for some incident vertex $w \not\in I^*$. It is straightforward to see that for $\varphi^*$ any signal $\sigma_v$ that has non-zero probability leads to an accept decision of the receiver if and only if $v \in I^*$. Moreover, the total probability that the receiver accepts is $|I^*| \cdot \frac{2}{|V| + 3|E|}$.
\qed

%
%

\subsection{Proof of Theorem~\ref{thm:verifyHard2}}
We again build a reduction from the \IS\ problem. Given a graph $G =(V,E)$, there is a signal $\sigma_v$ for every vertex $v \in V$. Moreover, for every vertex $v \in V$ there is a rejectable state $\theta_v$ with weight $|V|$ and an acceptable state $\theta'_v$ with weight $|V| - \deg(v)$. For both $\theta_v, \theta'_v$ we are forced to signal $\sigma_v$. For every edge $e = \{u,v\}$ there is an acceptable state $\theta_e$ with weight $1$, and in $\theta_e$ we can signal $\sigma_u$ or $\sigma_v$. The distribution $q$ assigns every state a probability proportional to its weight. The sum of all weights is $|E| + 2|V|^2 - \sum_v \deg(v) = 2|V|^2 - |E|$.

For any signal $\sigma_v$, the receiver will pick action A if and only if the signal is sent deterministically for all incident acceptable states, i.e., $\theta'_v$ as well as the $\deg(v)$ many states $\theta_e$ with $e = \{u,v\}$. Due to the construction, this implies that no two signals $\sigma_u$ and $\sigma_v$ for neighboring vertices in $G$ can simultaneously be accept signals. Consequently, the set of accept signals corresponds to an independent set of $G$. 

If $\sigma_v$ is an accept signal, the sender obtains a utility from this signal of $(2|V|)/(2|V|^2 - |E|)$. Hence, the utility of the sender is linear in the number of accept signals, i.e., proportional to the size of independent set. As such, the \classNP-hardness of approximation within a factor of $|V|^{1-\varepsilon} = n^{1-\varepsilon}$ for \IS\ applies.
\qed
\end{appendix}

\end{document}